\documentclass[11pt]{article}
\usepackage{amsmath}
\usepackage{amssymb}
\usepackage{amsfonts}
\usepackage{latexsym}
\usepackage{color}
\usepackage{graphicx}
 \usepackage{amsmath,mathtools}
\usepackage{graphicx,graphics,epstopdf}
\usepackage{subfigure}
\makeatletter
\@addtoreset{subfigure}{framenumber}
\makeatother
\catcode `\@=11 \@addtoreset{equation}{section}

\catcode `\@=12

\newcommand{\be}{\begin{equation}}
\newcommand{\en}{\end{equation}}
\newcommand{\bea}{\begin{eqnarray}}
\newcommand{\ena}{\end{eqnarray}}
\newcommand{\beano}{\begin{eqnarray*}}
\newcommand{\enano}{\end{eqnarray*}}
\newcommand{\bee}{\begin{enumerate}}
\newcommand{\ene}{\end{enumerate}}

\newcommand{\Hil}{{\cal H}}

\newcommand{\F}{{\cal F}}
\newcommand{\Lc}{{\cal L}}
\newcommand{\Sc}{{\cal S}}

\newcommand{\C}{{\cal C}}
\newcommand{\E}{{\cal E}}

\newcommand{\1}{1 \!\! 1}
\newtheorem{thm}{Theorem}

\newtheorem{lemma}[thm]{Lemma}
\newtheorem{prop}[thm]{Proposition}
\newtheorem{defn}[thm]{Definition}

\newenvironment{proof}{\noindent {\bf Proof:}}{\hfill$\Box$}

\begin{document}
\thispagestyle{empty}

\vspace*{.7cm}

\begin{center}
{\Large \bf Eigenvalues of non-hermitian matrices: a dynamical and an iterative approach. Application to a truncated Swanson model}   \vspace{2cm}\\

{\large F. Bagarello}
\vspace{3mm}\\[0pt]
Dipartimento di Ingegneria, \\[0pt]
Universit\`{a} di Palermo, I - 90128 Palermo,  and\\
\, INFN, Sezione di Napoli, Italy.\\[0pt]
E-mail: fabio.bagarello@unipa.it\\[0pt]
home page: www1.unipa.it/fabio.bagarello\\[0pt]

\vspace{4mm}

{\large F. Gargano}
\vspace{3mm}\\[0pt]
Dipartimento di Ingegneria, \\[0pt]
Universit\`{a} di Palermo, I - 90128 Palermo,  and\\
E-mail: francesco.gargano@unipa.it\\[0pt]

\vspace{4mm}

\end{center}

\begin{abstract}
We propose two different strategies to find eigenvalues and eigenvectors of a given, not necessarily Hermitian, matrix $A$. Our methods apply also to the case of complex eigenvalues, making the strategies interesting for applications to physics, and to pseudo-hermitian quantum mechanics in particular. 
We first consider a {\em dynamical} approach, based on a pair of ordinary differential equations defined in terms of the matrix $A$ and of its adjoint $A^\dagger$. Then we consider an extension of the so-called power method, for which we prove a fixed point theorem for $A\neq A^\dagger$ useful in the determination of the eigenvalues of $A$ and $A^\dagger$. The two strategies are applied to some explicit problems. In particular, we compute the eigenvalues and the eigenvectors of the matrix arising from a recently proposed quantum mechanical system, the {\em truncated Swanson model}, and we check some asymptotic features of the Hessenberg matrix.
\end{abstract}

\vspace{.7cm}

{\bf Keywords}:  Eigenvalues and eigenvectors for non-Hermitian operators; Truncated Swanson model; Hessenberg matrix.

\vfill

\newpage

\section{Introduction}

The problem of finding eigenvalues and eigenvectors of a given (finite-dimensional) matrix $A$ is very old and extremely relevant. The number of applications in which this is needed is countless. Just to cite a single one, the most relevant for us, eigenvalues of a quantum mechanical Hamiltonian $H$ are the energies allowed for the system $\Sc$ one is investigating, \cite{messiah}. In ordinary quantum mechanics, $H$ is taken to be Hermitian\footnote{All along this paper Hermitian and self-adjoint will be used as synonimous.}: $H=H^\dagger$, \cite{rs}. Here $H^\dagger$ is the (Dirac)-adjoint of $H$, \cite{baginbook}: in practice, if $H$ is an $n\times n$ matrix, $H^\dagger$ is simply the transpose and complex-conjugate\footnote{More in general, $H^\dagger$ is the operator satisfying the equality $\left<f,Hg\right>=\left<H^\dagger f,g\right>$, for all $f,g\in\Hil$, the Hilbert space where $\Sc$ lives, and $\left<.,.\right>$ is its scalar product.} of $H$. The Hamiltonian is not the only operator of $\Sc$ which is usually assumed to be Hermitian. Other operators usually can be associated to $\Sc$ with the same property: these are the so-called {\em observables} of $\Sc$. Clearly, the eigenvalues of all these operators are also necessarily real. Few decades ago, with the seminal paper \cite{ben1}, it became clear to the community of physicists (the mathematicians were already aware of this!) that reality of eigenvalues is implied by Hermiticity, but not vice-versa: non-Hermitian operators exist, both in finite and in infinite dimensional Hilbert spaces, whose eigenvalues are all real, see \cite{baginbagbook} and references therein for several examples. Moreover, it is now clear that some of the eigenvalues of a {\em physically relevant} (non Hermitian) Hamiltonian could easily be complex. This is what happens, for instance, when $PT$-symmetry is broken, \cite{benbook}, i.e. when the parameters of the Hamiltonian change in a proper way, see \cite{jli,guo} and references therein, or when some effective Hamiltonian is used to describe gain and loss effects. With this in mind it should be clear why we are so interested in finding possibly complex eigenvalues, and their related eigenvectors, of a given matrix, describing some physical characteristic of $\Sc$. Needless to say, many softwares exist in the market since many years which compute eigenvalues and eigenvectors of a given matrix. Some of them works mainly numerically (like Matlab). Others, like Mathematica, produce often analytical results. However other techniques may be relevant when the dimension of the matrices become larger and larger, as, for instances, in the case of the matrices involved in ranking web pages. In this case, power method looks more efficient, other than being mathematically quite interesting. We find also rather elegant the dynamical approach, considered in \cite{tang}, which associates a suitable differential equation to a matrix whose eigenvalues we need to compute. A review of strategies can be found in \cite{rev}.

The aim of this paper is to merge these two aspects, complexity of eigenvalues and mathematical strategies designed especially for large matrices, to propose new algorithms, based on a solid mathematical ground, to compute eigenvalues of  non Hermitian matrices. In Section \ref{sect2} we extend the dynamical approach analyzed in \cite{tang}. Section \ref{sec:fixedpoint} contains our {\em modified power method}. Applications  are discussed in Section \ref{sect4}. In particular, we propose an application in the realm of quantum mechanics to the finite-dimensional Swanson model, \cite{bag2018}. Conclusions are given in Section \ref{sect5}.

\section{Dynamical approach: introductory results for $A=A^\dagger$}\label{sect2}
In the first part of this section, for reader's convenience, we briefly review the approach proposed in \cite{tang}, while in the second part we will propose a possible extention of this approach, useful for our particular pourpouses.

Let $x\in\Hil=\mathbb{C}^n$ be an $n$-dimensional vector, and $A=A^\dagger$ an Hermitian matrix on  $\Hil$. Of course, since $A$ is Hermitian, its eigenvalues are necessarily real\footnote{In \cite{tang} the Hilbert space was taken to be $\mathbb{R}^n$. We prefer to adopt this more general choice here, also in view of our interest in quantum mechanics.}. We want to compute eigenvalues and eigenvectors of $A$ (or, at least, some of them).

The approach we discuss here is {\em dynamical}: we introduce a differential equation on $\mathbb{R}^n$, and we show that the solution of this equation converges to the eigenstate of $A$ corresponding to its highest eigenvalue, the so-called {\em leading eigenvector}. We start considering a vector $x(t)\in\mathbb{R}^n\subset \Hil$, depending on a continuous parameter $t\geq0$, and we assume $x(t)$ satisfies the following equation:
\be \frac{dx(t)}{dt}=-x(t)+f(x(t))=\|x(t)\|^2Ax(t)-\left<x(t),Ax(t)\right>x(t),
\label{21}\en where we have defined $f(x)=\left(\|x\|^2A+(1-\left<x,Ax\right>)\right)x$. It is clear that $f$ is a non-linear function from $\Hil$ into $\Hil$. 

\begin{defn}\label{def1}
	A vector $\xi(t)\neq0$ is an equilibrium point of (\ref{21}) if $f(\xi(t))=\xi(t)$, for all $t\geq0$.
\end{defn}
It is clear that if $\xi(t)$ is an equilibrium point of (\ref{21}), then $\xi(t)=\xi(0)=:\xi$ is an eigenstate of $A$ with eigenvalue $\lambda_\xi=\frac{\left<\xi,A\xi\right>}{\|\xi\|^2}$, and vice-versa. Notice that  $\lambda_\xi$ is well defined, since $\xi\neq0$, and therefore $\|\xi\|>0$. Moreover, since $A=A^\dagger$, $\lambda_\xi\in\mathbb{R}$.

Following \cite{tang} we look for a solution of (\ref{21}) in terms of the (unknown) eigenvectors of $A$: $Ae_j=\lambda_je_j$, $j=1,2,\ldots,n$: $x(t)=\sum_{i=1}^{n}c_i(t)e_i$. The equation for $c_i(t)$ can be deduced using the ortogonality of the $e_i$'s and the equality $\|x(t)\|^2=\|x(0)\|^2$, which can be proved easily since
\be
\frac{d\|x(t)\|^2}{dt}=\left<\dot x(t),x(t)\right>+\left<x(t),\dot x(t)\right>=0,
\label{add5}\en
after inserting twice (\ref{21}). We get
\be
\dot c_i(t)=F_i(t)c_i(t), \qquad F_i(t)=\|x(0)\|^2\lambda_i-D(t),
\label{22}\en
where $D(t)=\sum_{k=1}^{n}\lambda_k|c_k(t)|^2$. Notice that $F_i(t)$ is a real function of time. Rewriting the differential equation for $c_i(t)$ in its integral form, 
$
c_i(t)=e^{\int_{0}^{t}F_i(s)\,ds}c_i(0),
$
we see that, if $c_i(0)\in\mathbb{R}$, then $c_i(t)\in\mathbb{R}$ for all $t\geq0$. Working under this natural assumption we rewrite $D(t)=\sum_{k=1}^{n}\lambda_kc_k(t)^2$. The explicit form of the $c_i(t)$ can now be deduced and we get
\be
c_i(t)=\frac{c_i(0)\|x(0)\|}{\sqrt{\sum_{k=1}^{n}c_k(0)^2e^{2\|x(0)\|^2(\lambda_k-\lambda_j)t}}},
\label{23}\en
so that the general solution of (\ref{21}) is
\be x(t)=\sum_{i=1}^{n}\frac{c_i(0)\|x(0)\|}{\sqrt{\sum_{k=1}^{n}c_k(0)^2e^{2\|x(0)\|^2(\lambda_k-\lambda_j)t}}}\,e_i.
\label{24}
\en
Suppose now that the eigenvalues of $A$ are non degenerate. Hence we have $n$ different eigenvalues which we label as a decreasing sequence: $\lambda_1>\lambda_2>\cdots>\lambda_n$. Calling $B_{k,j}(t)=e^{2\|x(0)\|^2(\lambda_k-\lambda_j)t}$ we see the following: if $k>j$, $\lambda_k<\lambda_j$ and therefore $B_{k,j}(t)\rightarrow0$ when $t\rightarrow\infty$. If $k=j$, $\lambda_k=\lambda_j$ and therefore $B_{k,j}(t)=1$ for all $t$. Moreover, if $k<j$, $\lambda_k>\lambda_j$ and therefore $B_{k,j}(t)\rightarrow\infty$ when $t\rightarrow\infty$. With this in mind, and assuming for concreteness that $c_1(0)>0$, it follows that
$$
x(t)\rightarrow \|x(0)\|e_1, 
$$
when $t\rightarrow\infty$: the solution of the differential equation (\ref{21}) converges to the (non-normalized, in general) eigenvector of $A$ corresponding to its highest eigenvalue. This is true, of course, if $c_1(0)$ in
 (\ref{23}) and (\ref{24}) is different from zero. If we rather take, as initial vector $x(0)$, a vector with $c_1(0)=0$ and $c_2(0)\neq0$, $x(0)=\sum_{i=2}^{n}c_i(0)e_i$, we can check that
$$
x(t)\rightarrow \|x(0)\|e_2, 
$$
at least if $c_2(0)>0$. Incidentally, this result also shows that, if we consider a vector $x(0)$ which is orthogonal to $e_1$, then it remains orthogonal to $e_1$ even in the limit $t\rightarrow\infty$. These steps can be repeated and all the eigenvectors of $A$ can be found, in principle

The situation generalizes easily when some eigenvalue is degenerate. Suppose, just to be concrete, that $\lambda_1=\lambda_2>\lambda_3=\lambda_4>\lambda_5>\cdots>\lambda_n$: two eigenvalues are degenerate, and both have degeneracy two. In this case, repeating the previous analysis, we see that
$$
x(t)\rightarrow \|x(0)\|\frac{c_1(0)e_1+c_2(0)e_2}{\sqrt{c_1(0)^2+c_2(0)^2}}, 
$$
at least if both $c_1(0)$ and $c_2(0)$ are not zero. Notice that this vector belongs to the eigenspace of the highest eigenvalue of $A$, $\E_{max}$, as before, even if this eigenspace is now two-dimensional. The two leading eigenvectors can now be fixed by choosing any two orthonormal vectors in $\E_{max}$.

\subsection{Extension to  $A\neq A^\dagger$}\label{sect3}

Our next task is to generalize the above results to the case in which $A\neq A^\dagger$. As we have discussed in the Introduction, this is relevant in connection with pseudo hermitian quantum mechanics or for similar extensions of ordinary quantum mechanics, where the Hamiltonian of the physical system is not required to be Hermitian but still, most of the times, has a real spectrum, \cite{baginbook,benbook}. Also, it can be quite useful in other relevant situations, in which {\em physical} Hamiltonians turn out to have complex eigenvalues, like often happen in quantum optics or in gain-loss systems, see \cite{op1,op2,op3}. In this situation the first difference, with respect with what is discussed in Section \ref{sect2}, is that the eigenvectors of $A$, $\{\varphi_i\}$, are not, in general, mutually orthogonal: $\left<\varphi_i,\varphi_j\right>\neq\delta_{i,j}$. Nevertheless it is well known, \cite{bag2016,chri}, that a biorthogonal basis of $\Hil$ exists, $\{\Psi_i\}$, $\left<\varphi_i,\Psi_j\right>=\delta_{i,j}$, such that 
\bea
A\varphi_i=\lambda_i\varphi_i, \qquad \mbox{and}\qquad A^\dagger\Psi_i=\overline{\lambda_i}\Psi_i,\label{biorto}
\ena
for all $i=1,2,\ldots,n$. Let now introduce two sets of unknown functions $c:=\{c_i(t), \,i=1,2,\ldots,n\}$ and $d:=\{d_i(t), \,i=1,2,\ldots,n\}$, and two related vectors:
\be
x_\varphi(t)=\sum_{i=1}^{n}c_i(t)\varphi_i, \qquad x_\Psi(t)=\sum_{i=1}^{n}d_i(t)\Psi_i,
\label{31}\en
and let us assume  that these functions satisfy two (apparently) different differential equations:
\be
\left\{
\begin{array}{ll}
	\frac{dx_\varphi(t)}{dt}=\left<x_\Psi(t),x_\varphi(t)\right>Ax_\varphi(t)-\left<x_\Psi(t),Ax_\varphi(t)\right>x_\varphi(t)\\	
\frac{dx_\Psi(t)}{dt}=\left<x_\varphi(t),x_\Psi(t)\right>A^\dagger x_\Psi(t)-\left<x_\varphi(t),A^\dagger, x_\Psi(t)\right>x_\Psi(t).
\end{array}
\right.
\label{32}\en 
which look quite similar to two copies of equation \eqref{21}.
This {\em doubling} is a typical effect of going from Hermitian to non-Hermitian operators, \cite{baginbook}: when $A=A^\dagger$, then $\varphi_i=\Psi_i$, $\lambda_i=\overline{\lambda_i}$, for all $i$, and the eigenvectors form an orthonormal basis for $\Hil$.  Stated differently, moving from an Hermitian $A$ to a non-Hermitian one, often produces this kind of features: two Hamiltonians (one the adjoint of the other), two sets of eigenvalues, two families of eigenvectors, and so on, see \cite{baginbagbook}. 

The system in (\ref{32}) is closed and nonlinear. As we will show in a moment, it is possible to find its explicit solution, extending what we have shown for equation (\ref{21}).

It is interesting to notice that, similarly to what we have found in (\ref{add5}), i.e. that $\|x(t)\|=\|x(0)\|$, the two scalar products appearing in equations (\ref{32}) turn out to be time-independent:
\be
\left<x_\Psi(t),x_\varphi(t)\right>=\left<x_\Psi(0),x_\varphi(0)\right>=\sum_{k=1}^n \,\overline{d_k(0)}\, c_k(0)=\left<d,c\right>,
\label{34}\en
where $\left<d,c\right>$ is the usual scalar product of $d=\{d_k(0)\}$ and $c=\{c_k(0)\}$ in $\mathbb{C}^n$. The proof is a replica of that for $\|x(t)\|$:
$$
\frac{d}{dt}\left<x_\Psi(t),x_\varphi(t)\right>=\left<\dot x_\Psi(t),x_\varphi(t)\right>+\left<x_\Psi(t),\dot x_\varphi(t)\right>=0,
$$ 
after inserting (\ref{32}). Hence equations in (\ref{32}) can be rewritten as
\be
\left\{
\begin{array}{ll}
	\frac{dx_\varphi(t)}{dt}=\left<x_\Psi(0),x_\varphi(0)\right>Ax_\varphi(t)-\left<x_\Psi(t),Ax_\varphi(t)\right>x_\varphi(t)\\	
	\frac{dx_\Psi(t)}{dt}=\left<x_\varphi(0),x_\Psi(0)\right>A^\dagger x_\Psi(t)-\left<x_\varphi(t),A^\dagger, x_\Psi(t)\right>x_\Psi(t).
\end{array}
\right.
\label{32bis}\en

We introduce  the following definition:
\begin{defn}\label{def2}
	Two vectors $(\xi_\varphi(t),\xi_\Psi(t))\neq(0,0)$, $\xi_\varphi(t)=\sum_{i=1}^{n}c_i(t)\varphi_i$ and $\xi_\psi(t)=\sum_{i=1}^{n}d_i(t)\Psi_i$, form an equilibrium pair of (\ref{32bis}) if the right-hand sides of (\ref{32bis}) are both zero.
\end{defn}

Next we prove this lemma:

\begin{lemma}\label{lemma3}
	If $(\xi_\varphi(t),\xi_\Psi(t))$ is an equilibrium pair of (\ref{32bis}) then $(\xi_\varphi(t),\xi_\Psi(t))=(\xi_\varphi(0),\xi_\Psi(0))$, and, if $\left<\xi_\Psi(0),\xi_\varphi(0)\right>\neq0$,
	\be
	A\xi_\varphi(0)=\frac{\left<\xi_\Psi(0),A\xi_\varphi(0)\right>}{\left<\xi_\Psi(0),\xi_\varphi(0)\right>}\,\xi_\varphi(0), \qquad A^\dagger\xi_\Psi(0)=\frac{\left<\xi_\varphi(0),A^\dagger\xi_\Psi(0)\right>}{\left<\xi_\varphi(0),\xi_\Psi(0)\right>}\,\xi_\Psi(0). 
\label{33}\en
	Viceversa, if $(\xi_\varphi(t),\xi_\Psi(t))=(\xi_\varphi(0),\xi_\Psi(0))$, with $\left<\xi_\Psi(0),\xi_\varphi(0)\right>\neq0$, and if $(\xi_\varphi(t),\xi_\Psi(t))$ satisfy (\ref{33}), then $(\xi_\varphi(t),\xi_\Psi(t))$ is an equilibrium pair of (\ref{32bis}).
\end{lemma}

\begin{proof}
	Let us assume first that $(\xi_\varphi(t),\xi_\Psi(t))$ is an equilibrium pair of (\ref{32bis}). Hence the right-hand sides of (\ref{32bis}) are both zero, which implies that $\dot \xi_\varphi(t)=\dot \xi_\Psi(t)=0$. Then 
	$(\xi_\varphi(t),\xi_\Psi(t))=(\xi_\varphi(0),\xi_\Psi(0))$. Now, since for instance $\left<\xi_\Psi(0),\xi_\varphi(0)\right>A\xi_\varphi(t)-\left<\xi_\Psi(t),A\xi_\varphi(t)\right>\xi_\varphi(t)=0$, the first equality in (\ref{33}) easily follows recalling also that $\left<\xi_\Psi(0),\xi_\varphi(0)\right>\neq0$.
	
	Viceversa, if $(\xi_\varphi(t),\xi_\Psi(t))=(\xi_\varphi(0),\xi_\Psi(0))$, and if, for instance, $A\xi_\varphi(0)=\frac{\left<\xi_\Psi(0),A\xi_\varphi(0)\right>}{\left<\xi_\Psi(0),\xi_\varphi(0)\right>}\,\xi_\varphi(0)$, then $\left<\xi_\Psi(0),\xi_\varphi(0)\right>A\xi_\varphi(0)-\left<\xi_\Psi(0),A\xi_\varphi(0)\right>\xi_\varphi(0)=0$. Therefore $$\left<\xi_\Psi(0),\xi_\varphi(0)\right>A\xi_\varphi(t)-\left<\xi_\Psi(t),A\xi_\varphi(t)\right>\xi_\varphi(t)=0$$ as well, as we had to prove

\end{proof}

From formulas in (\ref{33})  we conclude that an equilibrium pair of (\ref{32}) is a set of time-independent eigenvectors of $A$ and $A^\dagger$ respectively, with eigenvalues given in terms of  $\xi_\varphi(t)$ and $\xi_\Psi(t)$, both computed at $t=0$ (or at any $t>0$, being these vectors constant in time).

Using now (\ref{34}), together with the expansions in (\ref{31}), and the biorthogonality of the sets $\{\varphi_i\}$ and $\{\Psi_i\}$, the equations in (\ref{32bis}) can be rewritten as follows:
\be
\left\{
\begin{array}{ll}
	\dot c_k(t)=\left(\left<d,c\right>\lambda_k-\Lambda(t)\right)c_k(t),\\	
	\dot d_k(t)=\left(\left<c,d\right>\overline{\lambda_k}-\overline{\Lambda(t)}\right) d_k(t),\\
\end{array}
\right.
\label{35}\en 
where we have defined
\be
\Lambda(t)=\sum_{k=1}^{n}\lambda_k\,\overline{d_k(t)}\,c_k(t).
\label{36}\en
It is convenient now to introduce two auxiliary functions $p_k(t)=e^{-\left<d,c\right>\lambda_kt}c_k(t)$ and   $q_k(t)=e^{-\left<c,d\right>\overline{\lambda_k}\,t}d_k(t)$. Hence we can rewrite (\ref{35}) as follows:
\be
\left\{
\begin{array}{ll}
	\dot p_k(t)=-\Lambda(t) p_k(t),\\	
	\dot q_k(t)=-\overline{\Lambda(t)} d_k(t).\\
\end{array}
\right.
\label{37}\en 
where now we rewrite $\Lambda(t)$ as follows:
$$
\Lambda(t)=\sum_{k=1}^{n}\lambda_k\,e^{2\left<d,c\right>\lambda_kt}\,\overline{q_k(t)}\,p_k(t).
$$
Next, let us consider the following initial conditions for $c_j(t)$ and $d_j(t)$: $d_j(0)=\overline{c_j(0)}$, for all $j=1,2,\ldots,n$. Notice that this choice automatically implies that $c$ and $d$ are not orthogonal, as required in Lemma \ref{lemma3}. Indeed we have $\left<d,c\right>=\sum_{k=1}^{n}\,(c_k(0))^2>0$ if we restrict, as we will do here, to real and not all zero $c_k(0)$. Hence $q_j(0)=\overline{p_j(0)}$, and the solutions for $p_k(t)$ and $q_k(t)$ are easily related: $q_k(t)=\overline{p_k(t)}$. The system in (\ref{37}) simplifies, giving a single equation
 \be
 	\dot p_k(t)=-\Lambda(t) p_k(t),
 \label{38}\en
 with
 $$
 \Lambda(t)=\sum_{k=1}^{n}\lambda_k\,e^{2\left<d,c\right>\lambda_kt}\,(p_k(t))^2,
 $$
and we also get $\left<d,c\right>=\sum_{k=1}^{n}\,(c_k(0))^2=\sum_{k=1}^{n}\,(p_k(0))^2$ which is strictly positive under our assumptions. From (\ref{38}) we deduce that for all $k$ and $l$ the ratio $\frac{p_k(t)}{p_l(t)}$ is independent of time:
$$
\frac{p_k(t)}{p_l(t)}=\frac{p_k(0)}{p_l(0)},
$$
at least if $p_l(0)=c_l(0)\neq 0$.
Using this property we obtain the following differential equation:
$$
\frac{d}{dt}\frac{1}{p_k^2(t)}=\frac{2}{p_k^2(0)}G(t),
$$
where $G(t)=\sum_{k=1}^{n}\lambda_k\,e^{2\left<d,c\right>\lambda_kt}\,( p_k(0))^2$ so that, after few simple computations and going back to $c_k(t)$, we obtain
$$
c_k(t)=\frac{|c_k(0)|\sqrt{\left<d,c\right>}}{\sqrt{\sum_{l=1}^{n}c_l(0)^2e^{2\left<d,c\right>(\lambda_l-\lambda_k)t}}}.
$$
Finally, from (\ref{31}), we get
\be x_\varphi(t)=\sum_{k=1}^{n}\frac{c_k(0)\sqrt{\chi_c(0)}}{\sqrt{\sum_{l=1}^{n}c_l(0)^2e^{2\chi_c(0)(\lambda_l-\lambda_k)t}}}\,\varphi_k.
\label{39}
\en
Similarly, $x_\Psi(t)=\sum_{k=1}^{n}\overline{c_k(t)}\,\Psi_k$

Now the next steps are almost identical to those for the Hermitian case: for instance, suppose that the eigenvalues of $A$ are non degenerate and that $\Re(\lambda_1)>\Re(\lambda_2)>\cdots>\Re(\lambda_n)$. Here $\Re(z)$ is the real part of $z$. Then, if $c_1(0)>0$, it follows that
$$
x_\varphi(t)\rightarrow \,\sqrt{\chi_c(0)}\varphi_1, \qquad x_\Psi(t)\rightarrow \,\sqrt{\chi_c(0)}\Psi_1,
$$
when $t\rightarrow\infty$: the solutions of the differential equations in (\ref{32}) converge to the (non-normalized, in general) eigenvectors of $A$ and $A^\dagger$ corresponding to the eigenvalue with the highest real part,  called again {\em the dominant eigenvectors}.

Of course, in order to determine the  eigenvectors $\varphi_i$ of $A$ and $\Psi_i$ of $A^\dagger$ related to the other $\lambda_i$, $i>1$,  it would be enough to start with initial conditions for \eqref{32} which are  orthogonal respectively to the eigenvectors $\Psi_1$ and $\varphi_1$, that is 
$\left<x_{\varphi}(0),\Psi_1\right>=\left<x_{\Psi}(0),\varphi_1\right>=0$ or, equivalently, taking  $c_1(0)=d_1(0)=0$. Hence, supposing that for instance $c_2(0)\neq0,d_2(0)\neq0$, we would have
$$
x_\varphi(t)\rightarrow \,\sqrt{\chi_c(0)}\,\varphi_2, \qquad x_\Psi(t)\rightarrow \,\sqrt{\chi_c(0)}\,\Psi_2.
$$
Similarly, by requiring that $c_1(0)=c_2(0),\ldots=c_k(0)=d_1(0)=d_2(0)\ldots=d_k(0)=0,$ we shall have $
x_\varphi(t)\rightarrow \,\sqrt{\chi_c(0)}\,\varphi_{k+1}$ and $x_\Psi(t)\rightarrow \,\sqrt{\chi_c(0)}\,\Psi_{k+1}$. We conclude that all the eigenvectors ca be found, in principle.

\vspace{2mm}

{\bf Remark:--}.
It is much simpler to recover the  eigenvalue with the smallest real part. In fact, it is enough to consider the system of equation \eqref{32bis}, and consequently \eqref{35}, with $A$ and $A^\dagger$ replaced respectively by $-A$ and $-A^\dagger$. In this case the eigenvalues  follow the order $-\Re(\lambda_n)>-\Re(\lambda_{n-1})>\ldots>-\Re(\lambda_1)$; hence the numerical solution in \eqref{39} will converge in the infinite time limit:
$$
x_\varphi(t)\rightarrow \,\sqrt{\chi_c(0)}\,\varphi_n, \qquad x_\Psi(t)\rightarrow \,\sqrt{\chi_c(0)}\,\Psi_n,
$$
corresponding to the lowest (in real part) eigenvalues $\lambda_n$  of $A$ and $A^\dagger$  (highest of $-A$ and $-A^\dagger$). 
\vspace{2mm}

 
%
%

\section{A fixed point strategy}\label{sec:fixedpoint}

In this section we consider a different approach for finding the eigensystem of a given matrix $A$, not necessarily self-adjoint, based on an iterative procedure. The proof of the convergence of the procedure is given introducing a suitable contraction map and proving that its unique fixed point is indeed one of the eigenvectors of $A$. Our strategy is based on some results discussed in \cite{rev,stak,schae}, which we extend here to cover the case of possibly non real eigenvalues of $A$.

The main idea of our fixed point approach is based on the following simple considerations: assume the $N\times N$ matrix $A$ admits $N$, possibly not all different, (complex) eigenvalues $\lambda_j$. In particular, we assume that
\be
|\lambda_1|>|\lambda_2|\geq |\lambda_3|\geq \cdots \geq |\lambda_N|.
\label{40}\en
This means, in particular, that $\lambda_1$ has multiplicity one. Using the same notation as in the previous section, this is called the { dominant eigenvalue}, and the related eigenvalue $u_1$ is the { dominant eigenstate}. The other eigenvalues can be degenerate, but still we can construct suitable linear combinations such that the set of eigenvectors of $A$, $\F_u=\{u_j\}$, $Au_j=\lambda_ju_j$, $j=1,2,\ldots,N$, is a basis for $\Hil=\mathbb{C}^N$. Of course, $\F_u$ is not, in general, an o.n. basis. However, as in Section \ref{sect3}, $\F_u$ admits a unique biorthogonal set $\F_v=\{v_j\}$, $\left<v_k,u_j\right>=\delta_{k,j}$, which is automatically a set of eigenstates of $A^\dagger$: $A^\dagger v_j=\overline{\lambda_j}\,v_j$, for all $j$.

Let now $x_0\in\Hil$, $\|x_0\|=1$, be such that $\left<v_1,x_0\right>\neq0$. Then $x_0=\sum_{i=1}^{N}\alpha_iu_i$, $\alpha_i=\left<v_i,x_0\right>$. In particular, $\alpha_1\neq0$. Now, it is clear that
\be
A^kx_0=\lambda_1^k\left(\alpha_1u_1+\sum_{i=2}^{N}\alpha_i\left(\frac{\lambda_i}{\lambda_1}\right)^ku_i\right),
\label{main40}
\en
for all $k=0,1,2,3,\ldots$. Since $\lambda_1$ is the dominant eigenvalue, $\left|\frac{\lambda_i}{\lambda_1}\right|<1$ for all $i=2,3,4,\ldots$, and the $k$-th power of this ratio goes to zero when $k\rightarrow\infty$. Of course, the larger the difference between $|\lambda_1|$ and $|\lambda_2|$, the faster the convergence of the sum $\sum_{i=2}^{N}\alpha_i\left(\frac{\lambda_i}{\lambda_1}\right)^ku_i$. If we further consider $x_k=\frac{A^kx_0}{\|A^kx_0\|}$, we deduce that, a part from corrections of orders $O\left(\left|\frac{\lambda_2}{\lambda_1}\right|^k\right)$,
$$
x_k\simeq \left(\frac{\lambda_1}{|\lambda_1|}\right)^k\frac{\alpha_1}{|\alpha_1|} \,u_1.
$$
Of course, this sequence converges if $\lambda_1>0$ . If $\lambda_1<0$, $\{x_k\}$ oscillates between two opposite vectors, both proportional to $u_1$, the dominant eigenvector. Of course, the sequence  $\{x_k\}$ {\em oscillates even more} if $\lambda_1$ is complex. It may be worth noticing that, in this procedure it is essential that $\alpha_1\neq0$. If we call $x$ the limit of $\left(\frac{|\lambda_1|}{\lambda_1}\right)^k\frac{|\alpha_1|}{\alpha_1}x_k$, $\lambda_1$ can be found as usual: $\lambda_1=\frac{\left<x,Ax\right>}{\left<x,x\right>}$.

{In complete analogy,  writing  $y_0=\sum_{i=1}^{N}\beta_iv_i$, $\beta_i=\left<u_i,y_0\right>$ the initial vector with $\beta_1\neq0$, we retrieve
\bea
y_k=\frac{\left(A^\dag\right)^ky_0}{\|\left(A^\dag\right)^ky_0\|}\simeq \left(\frac{\bar\lambda_1}{|\bar\lambda_1|}\right)^k\frac{\beta_1}{|\beta_1|} \,v_1,\label{secv1}
\ena which converge to an eigenvector proportional to $v_1$.}
\vspace{2mm}

{\bf Remarks:--} (1) Numerical implementation of this strategy clearly shows the effect of the sign of $\lambda_1$ which is positive when the sequence $\{x_k\}$ converges, while is negative when $\{x_k\}$ oscillates, for large $k$, between two opposite vectors, $x$ and $-x$.

(2) The same procedure can be easily extended to the case of a $d$-degenerate $\lambda_1$, $d>1$. In this case the sequence converges to an element of the $d$-dimensional subspace corresponding to the eigenvalue $\lambda_1$.

\vspace{2mm}

%

{The above strategy, which will be made rigorous soon, can be extended to find more eigenvectors others than the dominant ones, $u_1$ and $v_1$, by making use of the biorhonormal sets $\F_u$ and $\F_v$. In fact, once $u_1,v_1$ have been deduced, we consider a new {\em trial vector}, which we call again $x_0$, which is orthogonal to $v_1$: $\left<x_0,v_1\right>=0$. This implies that $x_0=\sum_{i=2}^{N}\alpha_i u_i$, $\alpha_i=\left<v_i,x_0\right>$. Hence, repeating the same steps as before, we deduce that
$$
x_k=\frac{A^kx_0}{\|A^kx_0\|}\simeq \left(\frac{\lambda_2}{|\lambda_2|}\right)^k\frac{\alpha_2}{|\alpha_2|} \,u_2,
$$
and all our previous considerations can be repeated. In particular, again we conclude that $x_k$ converges up to a phase to the second dominant eigenvector $u_2$.
 The procedure can be continued by finding $v_2$ as we did for $v_1$ in \eqref{secv1}, and  iterated more for all the eigenvalues and eigenvectors.
}
\subsection{The contraction}

Let us now define a map $T=\frac{1}{\lambda_1}\,A$, and let us fix a (normalized) $x_0\in\Hil$. We  define the following set:
$$
\C_{x_0}=\left\{f\in\Hil;\quad \left<v_1,f\right>=\left<v_1,x_0\right> \right\}.
$$

Of course, since both $\lambda_1$ and $v_1$ are unknown, when we start our procedure, both $T$ and $\C_{x_0}$ cannot be explicitly identified. However, as we will show in the rest of this section, they are useful tools to prove the convergence of the power method also in presence of complex eigenvalues.
From its definition we see that all the vectors in $\C_{x_0}$ have the same projections on $v_1$. It is clear that $T$ maps $\C_{x_0}$ into $\C_{x_0}$. In fact, let $f\in\C_{x_0}$. Then, expanding $f=\sum_{j=1}^{N}\left<v_j,f\right>u_j$, we have
$$
Tf=\frac{1}{\lambda_1}\sum_{j=1}^{N}\left<v_j,f\right>Au_j=\frac{1}{\lambda_1}\sum_{j=1}^{N}\left<v_j,f\right>\lambda_ju_j=\left<v_1,f\right>u_1+\frac{1}{\lambda_1}\sum_{j=2}^{N}\left<v_j,f\right>\lambda_ju_j,
$$ 
so that $\left<v_1,Tf\right>=\left<v_1,f\right>=\left<v_1,x_0\right>$.

Because of the fact that $A\neq A^\dagger$, the approach used in \cite{stak} does not work. This is because, in general, $\F_u$ is not an o.n. basis. Hence, if $f=\sum_{j=1}^Nf_ju_j$, $\|f\|^2\neq \sum_{j=1}^{N}|f_j|^2$. For this reason, we introduce a new norm $\|.\|_v$ in $\Hil$, which is more convenient for us. Notice however that, due to the fact that $\Hil$ is finite-dimensional, $\|.\|_v$ is equivalent to the standard norm $\|.\|$ of $\Hil$, $\|f\|=\sqrt{\left<f,f\right>}$. We put
\be
\|f\|_v=\sup_{j}|\left<v_j,f\right>|.
\label{41}\en
The fact that this is a norm is clear. In particular, $\|f\|_v=0$ if and only if $\left<v_j,f\right>=0$ for all $j$, which implies that $f=0$, due to the fact that $\F_v$, being a basis, is complete in $\Hil$. Since $\|.\|_v$ and $\|.\|$ are equivalent, and since $\C_{x_0}$ is a closed subspace of a complete set, $\C_{x_0}$ is also complete.  Our main result is contained in the following proposition:

\begin{prop}\label{prop::contraction}
$T$ is a contraction on $\C_{x_0}$. Hence it admits an unique fixed point $y_f\in\C_{x_0}$, $y_f=\left<v_1,x_0\right>u_1$.
\end{prop}

\begin{proof}
 Let $x,y\in\C_{x_0}$: $x=\sum_{j=1}^{N}\alpha_ju_j$,  $y=\sum_{j=1}^{N}\beta_ju_j$, $\alpha_j=\left<v_j,x\right>$ and $\beta_j=\left<v_j,y\right>$, with $\alpha_1=\beta_1$. Hence
 $$
 x-y=\sum_{j=2}^{N}(\alpha_j-\beta_j)u_j, \qquad Tx-Ty=\sum_{j=2}^{N}(\alpha_j-\beta_j)\left(\frac{\lambda_j}{\lambda_1}\right)u_j.
 $$
 Now, using (\ref{41}),
 $$
 \|x-y\|_v=\sup_{j}|\alpha_j-\beta_j|, \qquad \|Tx-Ty\|_v=\sup_{j}|\alpha_j-\beta_j|\left|\frac{\lambda_j}{\lambda_1}\right|\leq \left|\frac{\lambda_2}{\lambda_1}\right|\sup_{j}|\alpha_j-\beta_j|,
 $$
 so that, calling $\rho=\left|\frac{\lambda_2}{\lambda_1}\right|$, $\|Tx-Ty\|_v\leq \rho \|x-y\|_v$. Notice that (\ref{40}) implies that $\rho<1$. Hence $T$ is a contraction.
 
 The fact that $y_f=\left<v_1,x_0\right>u_1$ is a fixed point follows from a direct computation:
 $$
 Ty_f=\frac{1}{\lambda_1}\left<v_1,x_0\right>Au_1=\frac{1}{\lambda_1}\left<v_1,x_0\right>\lambda_1u_1=\left<v_1,x_0\right>u_1=y_f.
 $$

\end{proof}

\vspace{2mm}

{\bf Remarks:--} (1) It is clear that other fixed points of $T$ also exist: $\gamma u_1$, for all complex $\gamma$. This is because $T$ is a linear map. However, in $\C_{x_0}$, normalization of the fixed point cannot be changed, because of the condition $\left<v_1,f\right>=\left<v_1,x_0\right>$. This makes the fixed point of $T$ in $\C_{x_0}$ unique.

(2) It is interesting to extend this result to infinite-dimensional Hilbert spaces. This can be useful for possible applications to quantum mechanical systems living in $\Lc^2(\mathbb{R})$, the space of the square-integrable functions on $\mathbb{R}$, rather than only to elements of $\mathbb{C}^n$. But, as always when going from finite to infinite-dimensional vector spaces, mathematics is much more complicated. This is work in progress.

\vspace{2mm}

Once the fixed point $y_f$ is found, the related (dominant) eigenvalue is easily deduced:
\be
\frac{\left<y_f,Ay_f\right>}{\left<y_f,y_f\right>}=\frac{\left<u_1,Au_1\right>}{\left<u_1,u_1\right>}=\lambda_1,
\label{42}
\en as expected.

This fixed point strategy can be slightly modified to deduce more eigenvalues and eigenvectors other than the dominant ones. We will now briefly sketch what is known as the {\em shifted inverse power} method, modified to take into account the fact that, in our particular case, $A\neq A^\dagger$.

Let $q\in\mathbb{C}$ be a fixed complex number and suppose $q\neq\lambda_j$, for all $j$. This implies that $B_q=(A-q\1)^{-1}$ exists. Also, under our assumptions, $u_j$ is an eigenstate of $B_q$ with eigenvalue $\mu_j=(\lambda_j-q)^{-1}$: $B_qu_j=\mu_ju_j$. Let us call $j_0$ the value of the integer such that $|\mu_{j_0}|>|\mu_j|$, for all $j\neq j_0$. We now introduce the map $D=\frac{B_q}{\mu_{j_0}}$, and repeat for $D$ what we have done for $T$ before. In particular, it is possible to prove that $D$ has an unique fixed point $Y_f=\left<v_{j_0},x_0\right>u_{j_0}$ in the space $\C_{x_0}^{j_0}=\left\{f\in\Hil;\quad \left<v_{j_0},f\right>=\left<v_{j_0},x_0\right> \right\}$: $DY_f=Y_f$. Then, since
$$
\frac{\left<Y_f,DY_f\right>}{\left<Y_f,Y_f\right>}=1,
$$
we deduce that
$$
\frac{\left<Y_f,B_q^{-1}Y_f\right>}{\left<Y_f,Y_f\right>}=\mu_{j_0}=\frac{1}{\lambda_{j_0}-q},
$$
which gives back the value of $\lambda_{j_0}$ after few simple computations.

\subsection{The Schwartz Quotient}\label{sec::schwartz}
Proposition \ref{prop::contraction} in the previous section states that $T$ is a contraction on $\C_{x_0}$ and, as such, because of the properties of $\C_{x_0}$, it admits a unique fixed point. It is  worth stressing (once more!) that $T$ is defined by working as if the dominant eigenvalues $\lambda_1$ was known, which is not the case, of course: $\lambda_1$ will be computed only at the end of our numerical implementation. To avoid this apparent paradox, it can be useful to define a new (iterative)-map which  we shall prove  converges to $T$.
First of all we define the so called \textit{ Schwartz quotients},
\begin{equation}
\gamma_m=\frac{\left< \left(A^\dag\right)^{m+1} x_{\psi}, A^m x_{\varphi}\right>}{\left< \left(A^\dag\right)^{m} x_{\psi}, A^m x_{\varphi},\right>},\label{squot}
\end{equation}
where $x_\varphi=\sum_{j=1}^{N}\alpha_j u_j$ and $x_\psi=\sum_{j=1}^{N}\beta_jv_j$ are two initial random vectors with
$\alpha_j=\left<v_j,x_\varphi\right>$ and $\beta_j=\left<u_j,x_\psi\right>$.

It is clear that $\gamma_m$ is an approximation of the dominant eigenvalue $\lambda_1$ and it is expected to converge to it for $m\rightarrow\infty$. In fact, making use of the bi-orthogonality conditions of the eigenvectors of $A$ and $A^\dag$, we have that
\beano
\gamma_m=\lambda_1
\frac{\alpha_1\bar{\beta}_1+\sum_{j=2}^N\alpha_1\bar{\beta}_j\left(\frac{\lambda_j}{\lambda_1}\right)^{2m+1}}{
\alpha_1\bar{\beta}_1+\sum_{j=2}^N\alpha_1\bar{\beta}_j\left(\frac{\lambda_j}{\lambda_1}\right)^{2m}}=&&\\
\lambda_1\left(\frac{\alpha_1\bar{\beta}_1}{
\alpha_1\bar{\beta}_1+\sum_{j=2}^N\alpha_j\bar{\beta}_j\left(\frac{\lambda_j}{\lambda_1}\right)^{2m}}+
\frac{\sum_{j=2}^N\alpha_j\bar{\beta}_j\left(\frac{\lambda_j}{\lambda_1}\right)^{2m+1}}{
\alpha_1\bar{\beta}_1+\sum_{j=2}^N\alpha_j\bar{\beta}_j\left(\frac{\lambda_j}{\lambda_1}\right)^{2m}}\right)=&&\\
\lambda_1\left(\frac{1}{1+O\left(\frac{\lambda_2}{\lambda_1}\right)^{2m}}\right)+ O\left(\lambda_2\alpha_2\bar{\beta}_2\left(\frac{\lambda_2}{\lambda_1}\right)^{2m}\right).&&
\enano
Since $\lambda_1$ is the dominant eigenvalue, then $\left|\frac{\lambda_2}{\lambda_1}\right|<1$, and the $m$-th power of the ratio $\frac{\lambda_2}{\lambda_1}$ goes to zero when $m\rightarrow\infty$. Hence $\gamma_m\rightarrow\lambda_1$ for $m\rightarrow\infty$.

Then we define the map $\mathcal{S}_m=\frac{A}{\gamma_m} $ on 
$$
\C_{x_\varphi}=\left\{f\in\Hil;\quad \left<v_1,f\right>=\left<v_1,x_\varphi\right> \right\}.
$$
and we prove that  $\mathcal{S}_m$ can be used to define a sequence of vectors converging to the fixed point of $T$. This is not surprising since, recalling that $\gamma_m\rightarrow\lambda_1$, it is reasonable that $S_m\frac{A}{\gamma_m}\rightarrow\frac{A}{\lambda_1}$ when $m\rightarrow\infty$. To be more rigorous, we observe that 
$$
T^mx_\varphi=\alpha_1u_1+\sum_{j=2}^{N}\alpha_j\left(\frac{\lambda_j}{\lambda_1}\right)^{m}u_j,
$$
and
$$
S_mx_{m-1}=
\sum_{j=1}^{N}\alpha_j\left(\frac{\lambda_j^m}{\prod_{k=1}^m\gamma_k}\right)u_j,\\
$$
where $x_1=S_1x_\varphi,\ldots,x_m=S_mx_{m-1}$.
Hence
\beano
T^mx_\varphi-S_mx_{m-1}=\sum_{j=1}^{N}\alpha_j\left[\left(\frac{\lambda_j}{\lambda_1}\right)^m-\frac{\lambda_j^m}{\prod_{k=1}^m\gamma_k}\right]u_j.
\enano
and
\beano
\|T^mx_\varphi-S_mx_{m-1}\|_v=\sup_j\left|\alpha_j\left[\left(\frac{\lambda_j}{\lambda_1}\right)^m-\frac{\lambda_j^m}{\prod_{k=1}^m\gamma_k}\right]\right|.
\enano

Now, as $\gamma_m\approx\lambda_1+ O\left(\frac{\lambda_2}{\lambda_1}\right)^{2m}$,
\beano
\|T^mx_\varphi-S_mx_{m-1}\|_v=\sup_j\left|\alpha_j\left[\left(\frac{\lambda_j}{\lambda_1}\right)^m-\frac{\lambda_j^m}{\lambda_1^m+O\left(\left(\frac{\lambda_2}{\lambda_1}\right)^2\right)\lambda_1^{m-1}}\right]\right|=\\
\sup_j\left|\alpha_j\left[\left(\frac{\lambda_j}{\lambda_1}\right)^m-\frac{\lambda_j^m}{\lambda_1^m\left(1+O\left(\left(\frac{\lambda_2}{\lambda_1}\right)^2\right)\frac{1}{\lambda_1}\right)}\right]\right|\leq\\
\left|\left(\frac{\lambda_2}{\lambda
_1}\right)^m\right|\sup_j\left|\alpha_j\left(1-\frac{1}{1+O\left(\left(\frac{\lambda_2}{\lambda_1}\right)^2\right)\frac{1}{\lambda_1}}\right)\right|,
\enano
so that 
$$
\|T^mx_\varphi-S_mx_{m-1}\|_v\rightarrow0,
$$
for $m\rightarrow\infty$. Now, since $\lim_{m,\infty}T^mx_\varphi=u_1$, we can conclude that
also the sequence $S_mx_{m-1}\rightarrow u_1$ in the same limit.

\section{Numerical results}\label{sect4}

In this section we present some numerical  applications of the strategies proposed in Sections \ref{sect3} and \ref{sec:fixedpoint}.

The dynamical approach analyzed in Section \ref{sect3} requires the numerical solution of the system of ODEs \eqref{32bis}. In order to solve it, we have used 
a multi step variable order Adams-Bashforth-Moulton
time discretization method which  usually requires the solutions at several preceding time
points to compute the current solution, \cite{sg}. Once the solutions  $x_\varphi(t)$ and $x_\Psi(t)$ are computed at a specific time  $t$, the related {\em upgraded} eigenvalues are obtained as a consequence of the Lemma \ref{lemma3}, using the formula 
\be
	\lambda(t)=\frac{\left<x_\Psi(t),Ax_\varphi(t)\right>}{\left<x_\Psi(t),x_\varphi(t)\right>}\qquad \bar{\lambda}(t)=\frac{\left<x_\varphi(t),A^\dagger x_\Psi(t)\right>}{\left<x_\varphi(t),x_\Psi(t)\right>}.
\label{332}\en 
 Evaluation of the solution is then stopped when
\be
\|Ax_\varphi(t)-\lambda x_\varphi(t)\|<\delta_{tol},\quad \|A^\dag x_\Psi(t)-\bar{\lambda}x_\Psi(t)\|<\delta_{tol},
\label{add1}\en
where $\delta_{tol}$ is a small tolerance value. 

The fixed point approach, described in Section \ref{sec:fixedpoint}, makes use of the map $\mathcal{S}_m$ defined in 
Subsection \ref{sec::schwartz}. It  requires two initial guess vectors $x_\varphi=\sum_{j=1}^{N}\alpha_j u_j$ and $x_\psi=\sum_{j=1}^{N}\beta_jv_j$, and at the generic iteration $k$ the eigenvalue approximation ($\lambda_k$) used to define $\mathcal{S}_k$ is given by the Schwartz quotient \eqref{squot}.
Then the upgraded eigenvector of $A$ is given by $x_k=\mathcal{S}_kx_{k-1}$, where $x_0=x_\varphi$. Iteration is stopped when, as in (\ref{add1}),   
\be
\|Ax_k-\lambda_k x_k\|<\delta_{tol}.
\label{add2}\en

\subsection{A test case:  $7\times7$ matrix }\label{sectIV1}
The first numerical experiment (E1), deals with a non hermitian random squared matrix of order 7 of the form
\be
A=R+iT,
\label{add3}\en
where $R$ and $T$ are the random matrices
$$
R=\left(
\begin{array}{ccccccc}
 0.445 & -0.219 & 0.489 & 0.770 & 0.589 & -0.00333 & 0.950 \\
 0.481 & -0.892 & -0.806 & -0.743 & -0.641 & -0.422 & 0.701 \\
 -0.735 & 0.747 & 0.750 & -0.879 & 0.884 & -0.0114 & -0.260 \\
 0.528 & 0.357 & 0.707 & 0.986 & 0.201 & 0.320 & 0.207 \\
 0.899 & 0.727 & 0.206 & -0.792 & 0.109 & 0.895 & 0.672 \\
 -0.400 & -0.259 & -0.988 & 0.459 & 0.681 & 0.843 & 0.788 \\
 0.326 & -0.530 & -0.168 & 0.141 & 0.0158 & -0.496 & -0.907 \\
\end{array}
\right)
$$

$$
T=\left(
\begin{array}{ccccccc}
 0.959 & 0.314 & -0.237 & 0.232 & -0.608 & 0.199 & -0.164 \\
 -0.744 & -0.112 & 0.239 & 0.384 & -0.132 & 0.299 & 0.817 \\
 0.921 & 0.681 & -0.302 & 0.942 & -0.781 & 0.908 & -0.0566 \\
 0.465 & -0.641 & 0.505 & -0.892 & -0.830 & 0.715 & -0.170 \\
 -0.807 & 0.978 & -0.185 & -0.619 & -0.923 & 0.322 & -0.690 \\
 0.0672 & 0.893 & 0.620 & 0.711 & -0.631 & -0.636 & -0.211 \\
 -0.742 & -0.0257 & 0.536 & -0.952 & -0.325 & 0.0701 & 0.196 \\
\end{array}
\right).$$
The eigenvalues of $A$, as deduced by Matlab, for instance, and ordered from those having largest to those with lowest real parts, are 
$\lambda_1=1.5181 - 1.2564i, \lambda_2=0.9604 - 2.2206i,\lambda_3=0.9394 - 0.6078i,\lambda_4=0.8326 + 2.0418i,\lambda_5= -0.7583 - 1.154,\lambda_6= -0.8380 + 0.1978i,\lambda_7=-1.3201 + 1.2896i$. In what follows we will show how these values can be recovered using our strategies.
\subsubsection{Dynamical approach, experiment E1}
Numerical solution of \eqref{32bis} has been carried by fixing the tolerance on the convergence of the solution to
$\delta_{tol}=10^{-8}$.  Starting with a random initial conditions, the obtained converging solutions  $x_{\varphi_1}^\infty$ and $x_{\Psi_1}^\infty$, related to the first eigenvalues $\lambda_1,\bar\lambda_1$, are 
  $$x_{\varphi_1}^\infty=\left( \begin{array}{l}
   0.1451 + 0.0193i\\
  -0.1616 + 0.0317i\\
  -0.2025 - 0.0594i\\
   0.1936 + 0.1459i\\
   0.4491 - 0.2830i\\
   0.6958 - 0.2368i\\
   0.0554 - 0.1155i
   \end{array}\right),\quad 
   x_{\Psi_1}^\infty=\left( \begin{array}{l}
        0.0387 + 0.0163i\\
        0.0072 - 0.0091i\\
       -0.0274 - 0.0008i\\
        0.0256 - 0.0002i\\
        0.0234 + 0.0520i\\
        0.0452 + 0.0044i\\
        0.0317 + 0.01\\
      \end{array}\right).
   $$
   
   The related eigenvalues, obtained trough \eqref{332}, are $\lambda_1=1.5181 - 1.2564i$ and $\bar{\lambda}_1=1.5181 + 1.2564i$. Notice that  $x_{\varphi_1}^\infty$ and $x_{\Psi_1}^\infty$ are not automatically bi-normalized, since $\left<x_{\varphi_1}^\infty,x_{\Psi_1}^\infty\right>=0.0415 + 0.0457i\neq 1$,
   and hence they represents  the bi-orthonormal vectors $\varphi_1$ and $\Psi_1$ in \eqref{biorto}, respectively, only up to some normalizations. 
   In Figure \ref{fig:7x7}(a) the convergence of the absolute values of the components of the dynamical solution $x_{\varphi_1}(t)$  to $x_{\varphi_1}^\infty$ is shown (we find a similar behaviour for convergence to $x_{\Psi_1}^\infty$, not shown  in figure). Real and imaginary parts of the eigenvalue $\lambda_1$ obtained at each time with 
   \eqref{332}, are shown in
   Figure \ref{fig:7x7}(c).
   
   Concerning the determination of the eigenvectors related the eigenvalue $\lambda_7$ having the smallest real part, we have solved \eqref{32bis} by replacing  $A$ and $A^\dagger$ with $-A$ and $-A^\dagger$ respectively.
   The  solutions $x_{\varphi_7}^\infty$ and $x_{\Psi_7}^\infty$  are
     $$x_{\varphi_7}^\infty=\left( \begin{array}{l}
        -0.1751 + 0.1561i\\
        -0.5680 + 0.0590i\\
        -0.0281 + 0.1148i\\
         0.1124 - 0.0007i\\
         0.1513 + 0.4898i\\
        -0.2348 + 0.1173i\\
        -0.2529 - 0.4435
      \end{array}\right),\quad 
      x_{\Psi_7}^\infty=\left( \begin{array}{l}
         -0.00032 + 0.00015\\
          -0.00033 + 0.00089\\
          -0.0005 + 0.00025\\
          -0.00011 + 0.00016i\\
           0.000792 - 0.000228i\\
          -0.00055 + 0.00044i\\
          -0.001379 - 0.00015\\
         \end{array}\right).
      $$
   The corresponding eigenvalues are $\tilde\lambda_7=1.3201 - 1.2896i$ and $\bar{\tilde\lambda}_7=1.3201 + 1.2896i$  which, as expected, are related to the eigenvalues $\lambda_7=-1.3201 + 1.2896i$ and $\bar\lambda_7=-1.3201 - 1.2896i$ of $A$ and $A^\dagger$ by a simple change of sign.
   Again, we see that the solutions we get are  not automatically bi-normalized,  since $\left<x_{\varphi_7}^\infty,x_{\Psi_7}^\infty\right>=0.00095 - 0.0014i\neq1$, so that $x_{\varphi_7}^\infty$ and $x_{\Psi_7}^\infty$ coincide with $\varphi_7$ and $\Psi_7$ only up to some normalization factor. 
   However the bi-orthogonality conditions with the vectors $x_{\varphi_1,\Psi_1}^\infty$ are  satisfied since
   the various scalar products $\left<x_{\varphi_1,\Psi_1}^\infty,x_{\varphi_7,\Psi_7}^\infty\right>$
   are well below the tolerance imposed (in general of order the precision of the machine $10^{-15}$): $\left<x_{\varphi_1,\Psi_1}^\infty,x_{\varphi_7,\Psi_7}^\infty\right>\simeq 0$.
\begin{figure}[!ht]
	\begin{center}		
		\hspace*{-2.0cm}\subfigure[Experiment E1: time evolutions of the components of $|x_{\varphi_1}(t)|$ obtained from \eqref{32}]{\includegraphics[width=8.5cm]{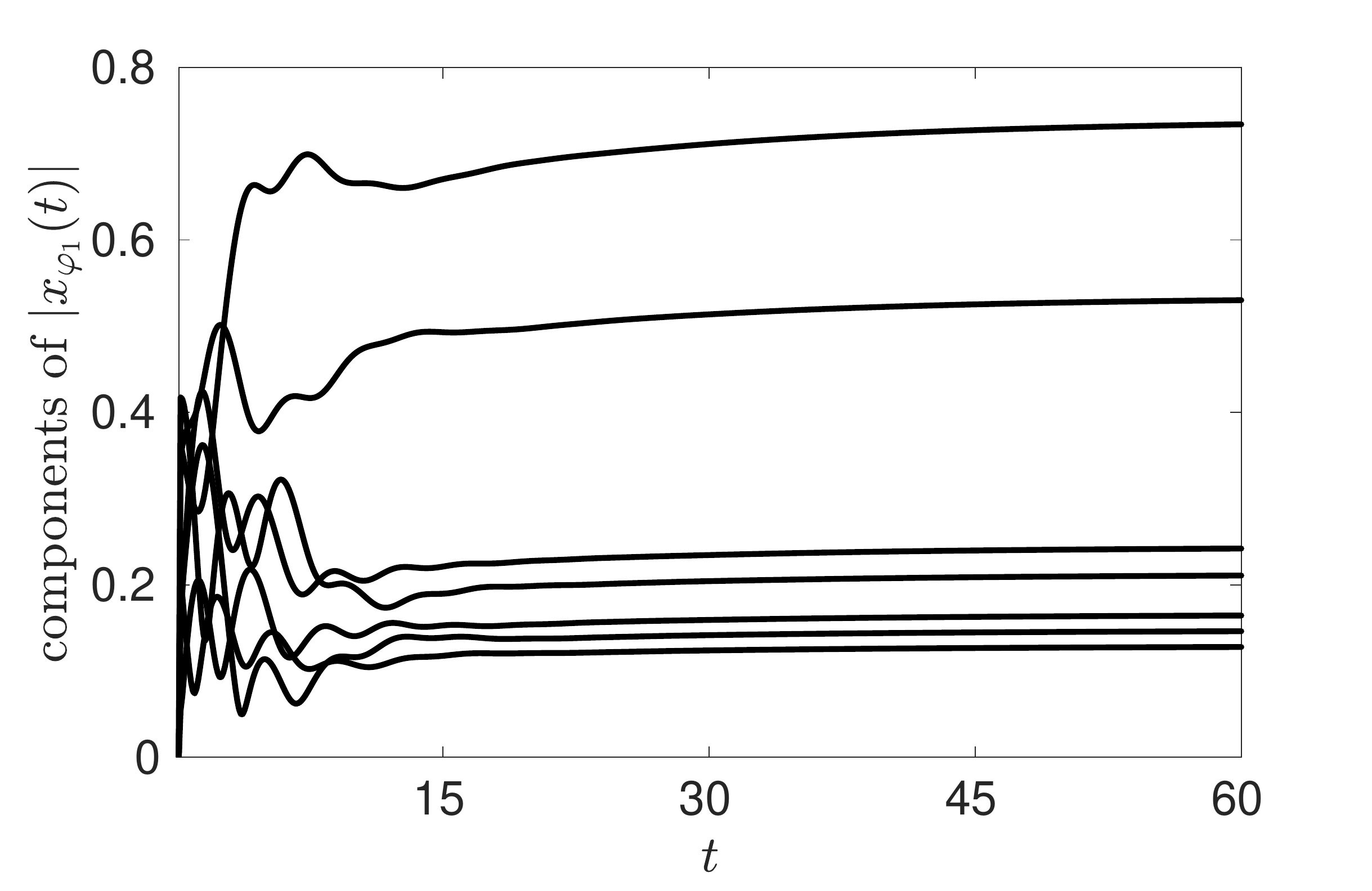}}
		\hspace*{+1.0cm}\subfigure[Experiment E1: time evolutions of the $|x_{\varphi_7}(t)|$ obtained from \eqref{32} with the matrices $-A,-A^\dagger$]{\includegraphics[width=8.5cm]{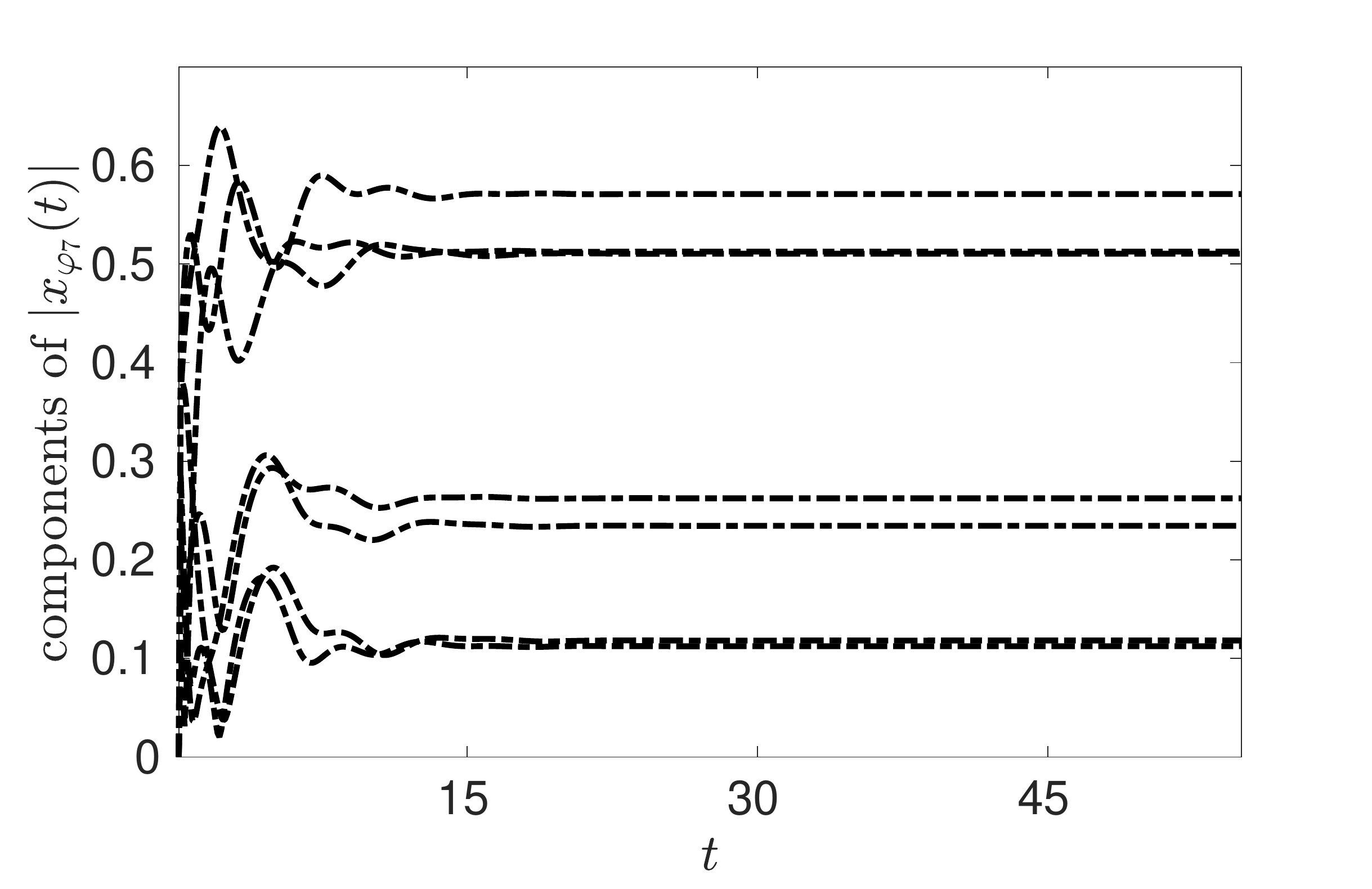}}
			\hspace*{-0.4cm}\subfigure[Experiment E1: time evolutions of the real and imaginary parts of the eigenvalues $\lambda_1$ and $\lambda_7$ of $A$]{\includegraphics[width=8.5cm]{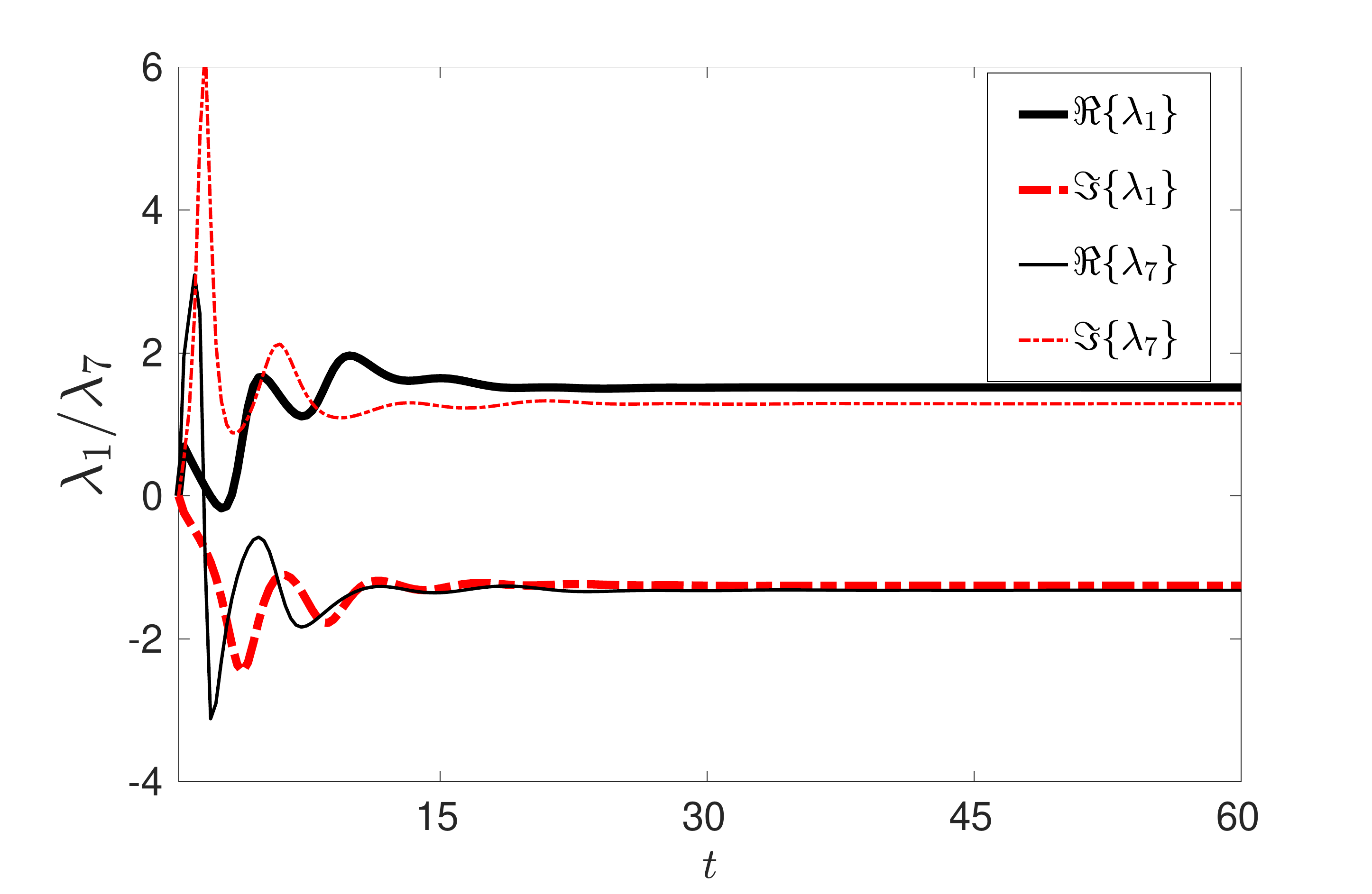}}
		\caption{Numerical results of dynamical approach applied to the experiment E1. }
		\label{fig:7x7}
	\end{center}
\end{figure}
\subsubsection{Fixed point approach, experiment E1}
In this subsection we show the results concerning the fixed point strategy applied to the same $7\times7$ matrix $A$ as in (\ref{add3}). Starting with a random initial condition $x_0$ and a tolerance $\delta_{tol}=10^{-8}$, the first set of iterations converge toward  the vector 
$$\tilde\varphi_2=\left( \begin{array}{l}
   -0.1366 - 0.2000i\\
    0.2704 + 0.0788i\\
   -0.5370 + 0.3358i\\
   -0.1461 - 0.2711i\\
   -0.4839 + 0.0951i\\
    0.2303 - 0.0502i\\
   -0.1436 - 0.2160i
           \end{array}\right).
$$
We call this vector $\tilde\varphi_2$ since it corresponds to the eigenvalue  $\lambda_2=0.9604 - 2.2206i$ of $A$ given above, which is the  largest in norm.
The convergence of the initial guess to the first eigenvector is shown in Figure \ref{fig:7x7fp}(a), where the norm of the components of the vector are shown. Convergence of the real and imaginary parts of the eigenvalue is shown in Figure \ref{fig:7x7fp}(d).
By applying the {shifted inverse power} method, by picking randomly complex values $q$,
we then obtain the other eigenvectors. The first two of them are in sequence 
$$\tilde\varphi_5=\left( \begin{array}{l}
        -0.3639 - 0.0758i\\
          0.3899 - 0.3802i\\
         -0.2303 + 0.3945i\\
          0.2895 + 0.3292i\\
          0.0127 + 0.2107i\\
         -0.2879 - 0.181\\
         -0.0623 + 0.0175i
      \end{array}\right),\quad 
      \tilde\varphi_7=\left( \begin{array}{l}
            0.2049 - 0.1141i\\
            0.5672 + 0.0665i\\
            0.0524 - 0.1059i\\
           -0.1099 - 0.0238i\\
           -0.0407 - 0.5110i\\
            0.2547 - 0.0633i\\
            0.1500 + 0.4880i
         \end{array}\right),
      $$
corresponding to the eigenvalues $\lambda_5$ and $\lambda_7$. Notice that with the shifted inverse power method the obtained sequence of eigenvectors do not follow the norm ordering (in fact $\lambda_5$ and $\lambda_7$ are not the greatest eigenvalues in norm after $\lambda_2$). The convergence  to the  eigenvectors $\tilde\varphi_5$ and $\tilde\varphi_7$ is shown in Figure \ref{fig:7x7fp}(b-c), where the norm of the components of the vectors are shown. Convergence of the real and imaginary parts of the corresponding eigenvalues is shown again in Figure \ref{fig:7x7fp}(d).

	\begin{center}	     
\begin{figure}[!ht]
	
		\hspace*{-2.0cm}\subfigure[Experiment E1: absolute values of the components of the first eigenvector
		of $A$]{\includegraphics[width=8.5cm]{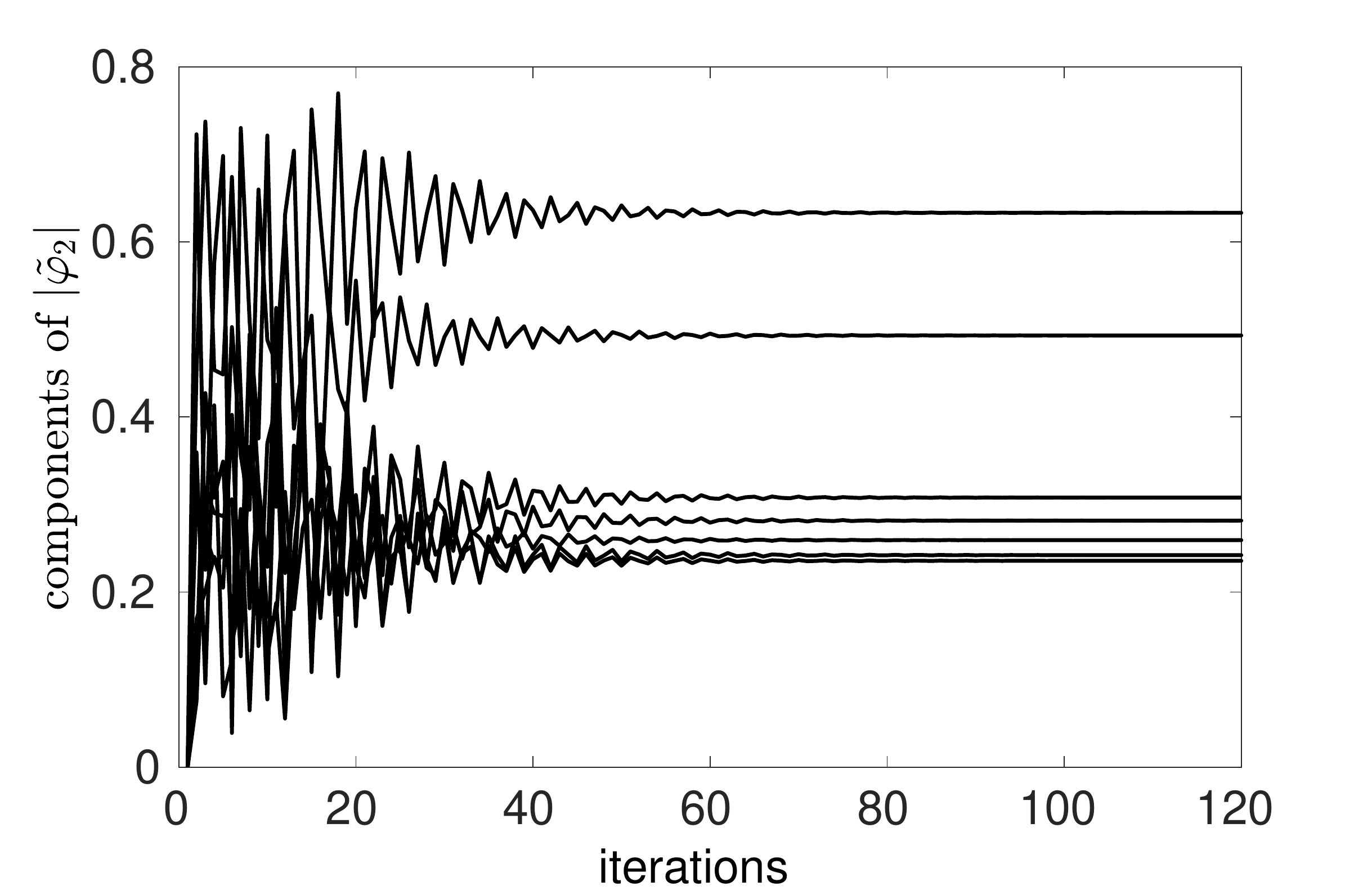}}
		\hspace*{+1.0cm}\subfigure[Experiment E1: absolute values of the components of the second eigenvector
				of $A$]{\includegraphics[width=8.5cm]{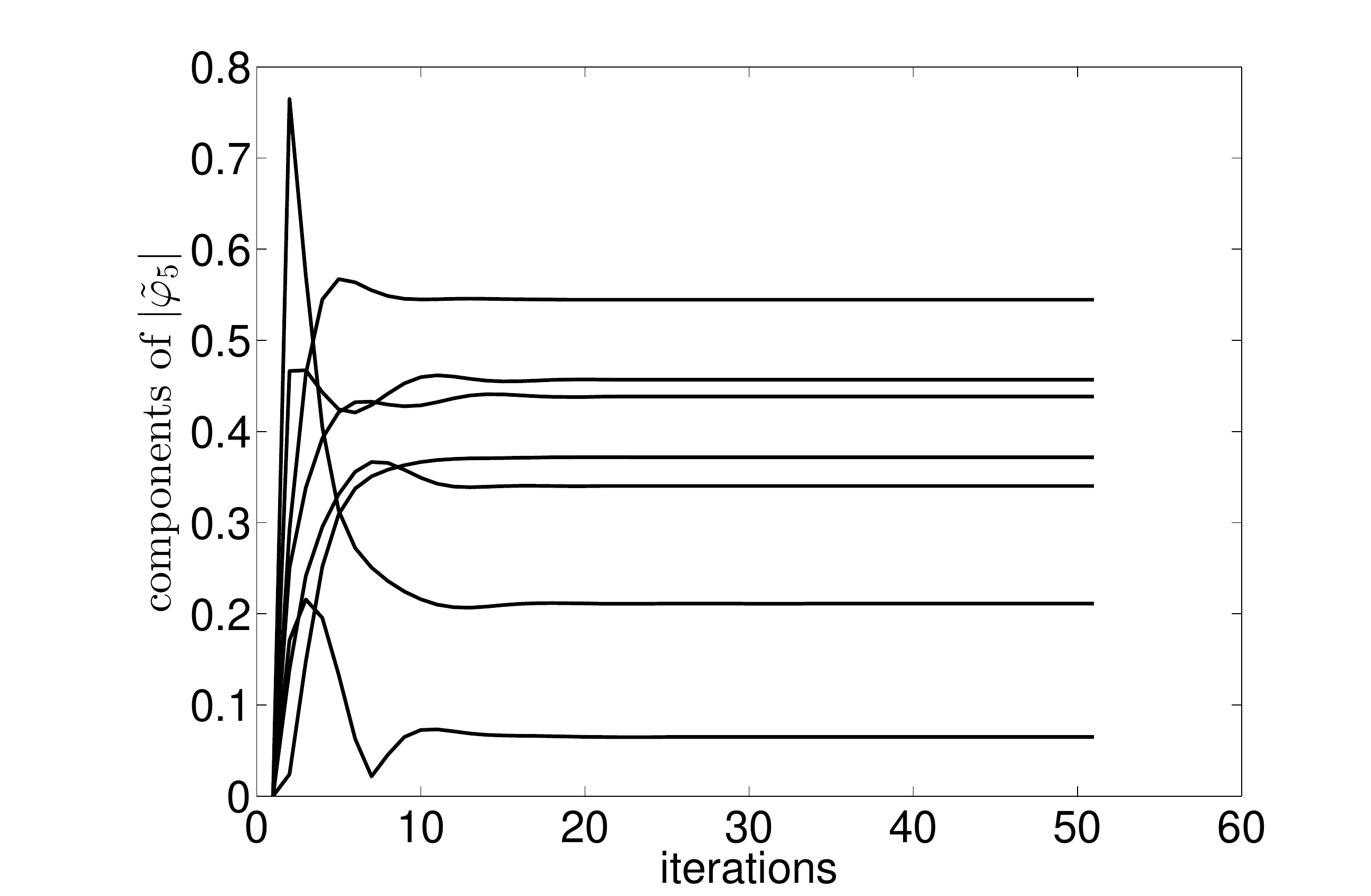}}\\
			\hspace*{-2.0cm}\subfigure[Experiment E1: absolute values of the components of the third eigenvector
							of $A$]{\includegraphics[width=8.5cm]{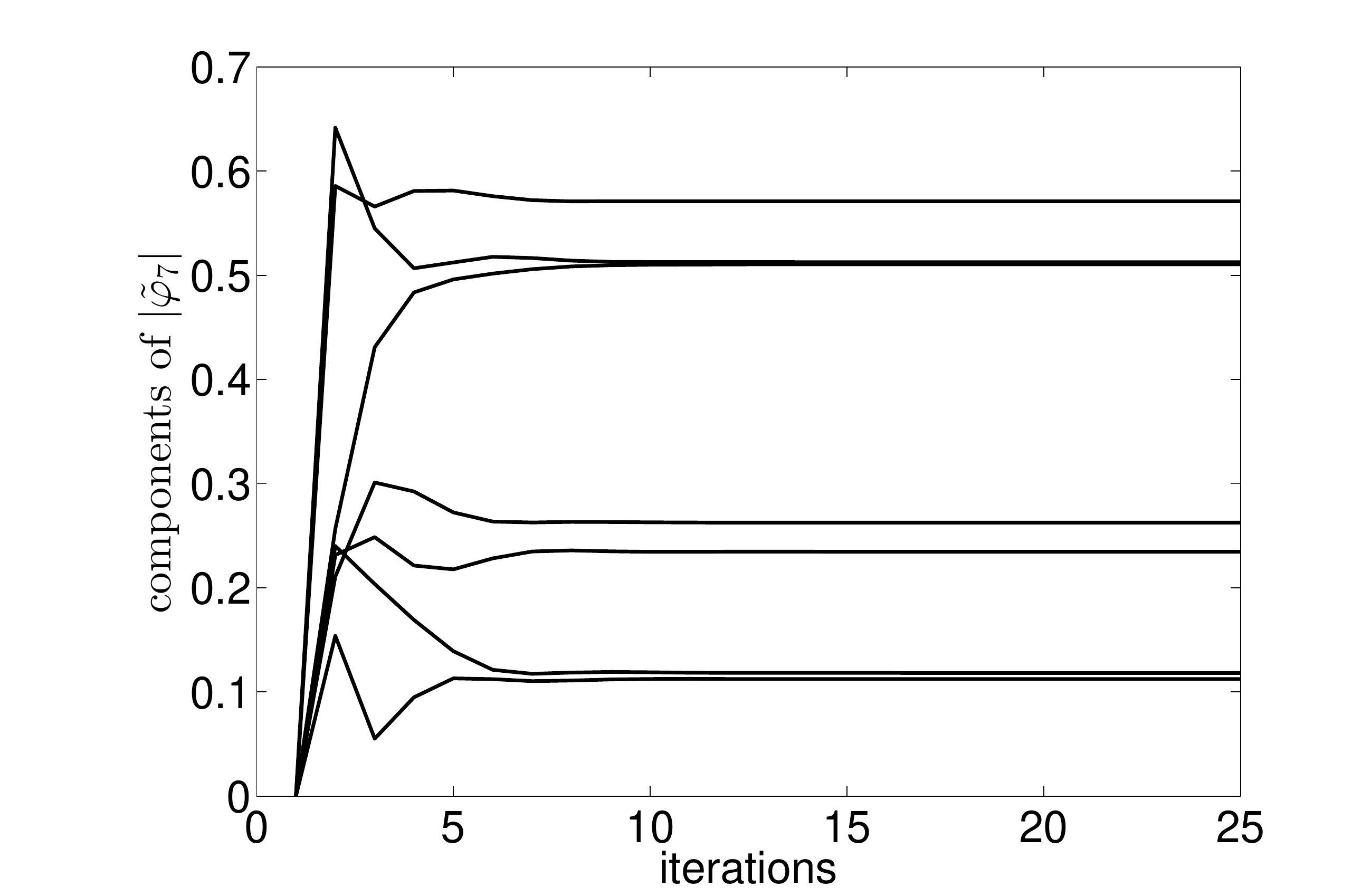}}
			\hspace*{+1cm}\subfigure[Experiment E1:convergence of real and imaginary parts of the three eigenvalues  $\lambda_2,\lambda_5$ and $\lambda_7$]{\includegraphics[width=8.0cm]{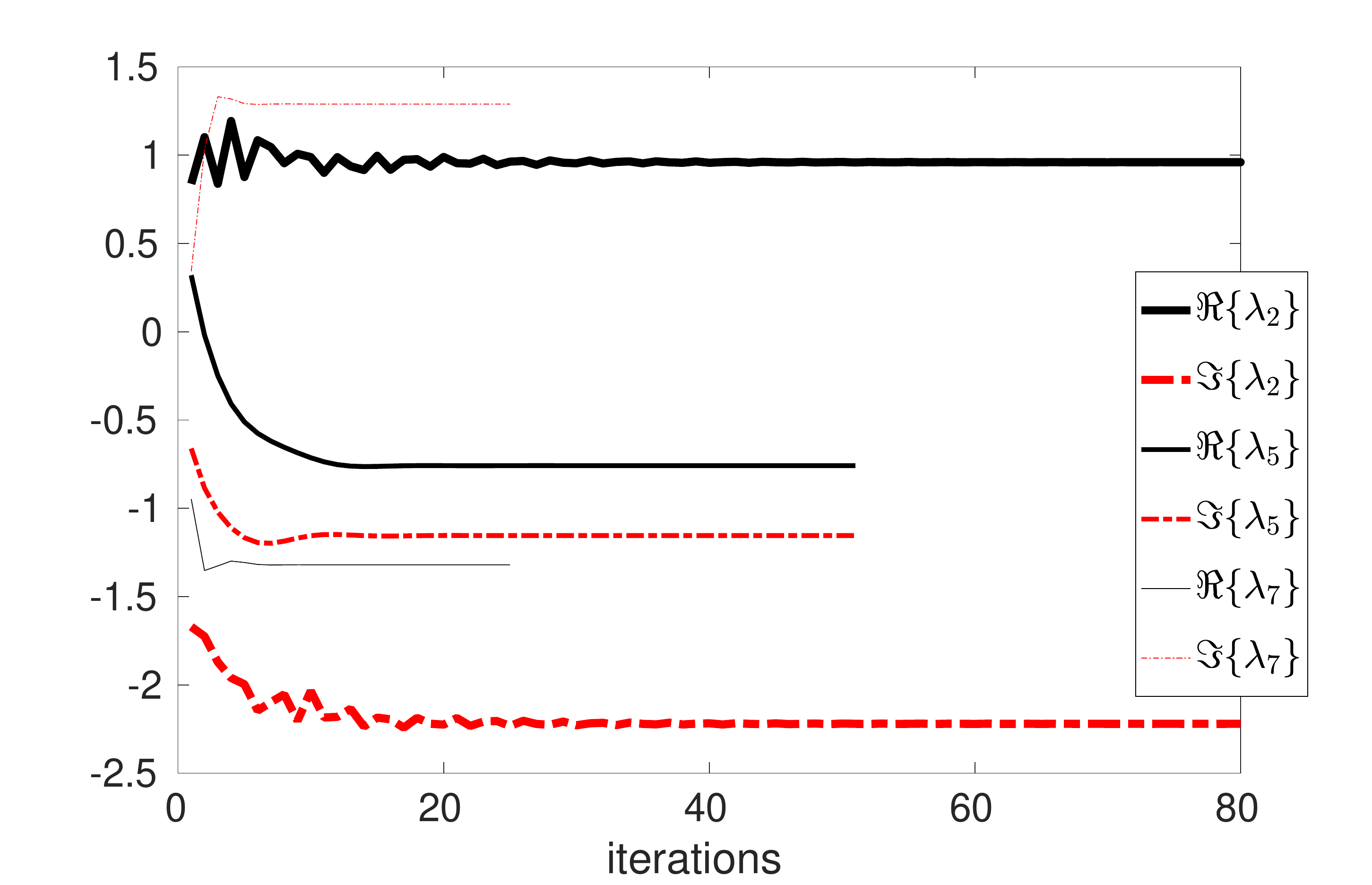}}
		\caption{Numerical results of fixed point approach applied to the experiment E1. }
		\label{fig:7x7fp}
\end{figure}
	\end{center}

\subsection{The Hessenberg matrix}\label{sectIV2}

In this second experiment (E2) we apply our numerical procedures to the finite Hessenberg matrix, an upper matrix with positive subdiagonal. The procedure to construct this kind of matrix is well known and we refer to \cite{Escr,Escr2} for more details.

The finite Hessenberg Matrix of order $n$ we consider is fully defined by a sequence $\{\alpha_k\}_{k\in\mathbb{N}}$ such that $\lim\limits_{k\rightarrow\infty}\alpha_k=0$:
$$
D\{\alpha_k\}=\begin{pmatrix}
-\alpha_1\bar{\alpha_0}  & -\frac{k_0}{k_1}\alpha_2\bar{\alpha}_0 & -\frac{k_0}{k_2}\alpha_3\bar{\alpha}_0 & \ldots\ & -\frac{k_0}{k_{n-1}}\alpha_{n}\bar{\alpha}_0\\ 
\frac{k_0}{k_1} & -\alpha_2\bar{\alpha}_1  & -\frac{k_1}{k_2}\alpha_3\bar{\alpha_1}   &\ldots  &  -\frac{k_1}{k_{n-1}}\alpha_{n}\bar{\alpha}_1\\ 
0 &  \frac{k_1}{k_2} & -\alpha_3\bar{\alpha_2}  & \ldots  & -\frac{k_2}{k_{n-1}}\alpha_{n}\bar{\alpha}_2 \\ 
\vdots & \vdots & \vdots & \vdots &\vdots \\ 
0 & 0 & 0 & \ldots & -\alpha_n\bar{\alpha}_{n-1}&
\end{pmatrix},
$$
with the various $k_j$  related to $\{\alpha_k\}$ as follows: $k_0=1, k_j=k_{j-1}/\sqrt{1-|\alpha_k|^2}$. 
{It is well known that the elements of $D\{\alpha_k\}$ converges for large $n$ to those of $S_R$, the (finite) right shift matrix, see \cite{Escr2}. Here $S_R=\{\delta_{j+1,j}\}_{j\in{1,\ldots,n}}$. It can be seen that the faster the sequence $\{\alpha_k\}_{k\in\mathbb{N}}$ converges to zero,  the faster the eigenvalues of  $D\{\alpha_k\}$ converge also to zero.}

\subsubsection{Dynamical approach, experiment E2}
{We start our analysis on $D\{\alpha_k\}$ by first solving  \eqref{32bis} with $A=D\{\alpha_k\}$, with the size of matrix $n=15$, a variable  tolerance (depending on the eigenvalue considered\footnote{This is important since the eigenvalues can be very small, as already discussed.}) 
from $\delta_{tol}=10^{-10}$ to $\delta_{tol}=10^{-14}$, and  random initial conditions.
We stress that, for this experiment, due to the particularly small values of the eigenvalues, we need to use smaller tolerance than in experiments E1 above and E3 below.}

Fixing $\delta_{tol}=10^{-10}$, the convergence to the determined eigenvalues with the largest and smallest real part, $\lambda_1$ and $\lambda_{15}$ respectively, are shown in Figs.\ref{fig:Hess_dyna}-\ref{fig:Hess_dynb} respectively for the cases $\{\alpha_k\}_{k\in\mathbb{N}}=\{\exp(-k^2)\}_{k\in\mathbb{N}}$ and
$\{\alpha_k\}_{k\in\mathbb{N}}=\{1/(k^2!)\}_{k\in\mathbb{N}}$: as previously explained the eigenvalue  $\lambda_{15}$ is easily retrieved by replacing  $A=D\{\alpha_k\}$ with $A=-D\{\alpha_k\}$ in \eqref{32bis}.

{As expected the solutions we obtain converge to eigenvalues which are very small and approach the value $0$: for the matrix $D\{\exp(-k^2)\}$ of order $n=15$, we obtain $\lambda_1=0.0059$
and $\lambda_{15}=-0.0233$, whereas for the matrix $D\{(1/(k^2!)\}$ of order $n=15$, the converging values are $\lambda_1=6.58\cdot 10^{-5}$ and $\lambda_{15}=-0.0417$.

To retrieve the other eigenvalues we have to  decrease further the tolerance to 
$\delta_{tol}=10^{-14}$ (close to the minimal resolution allowed by the machine) as the other eigenvalues  decrease to zero very rapidly: for instance the subsequent eigenvalues retrieved tend to the values $1.81\cdot10^{-8}$ and $1.32\cdot 10^{-11}$, which of course require a very small tolerance to be well determined.
}

\begin{figure}[!ht]
	\begin{center}	
		\hspace*{-2cm}\subfigure[Experiment E2: convergence of real and imaginary parts of the eigenvalues having the largest  ($\lambda_1$) and the smallest ($\lambda_{15}$) real part obtained with \eqref{32}
		for the Hessenberg matrix $D\{\exp(-k^2)\}$ of order $n=15$. In the inset the early time evolution.]
		{\includegraphics[width=8.5cm]{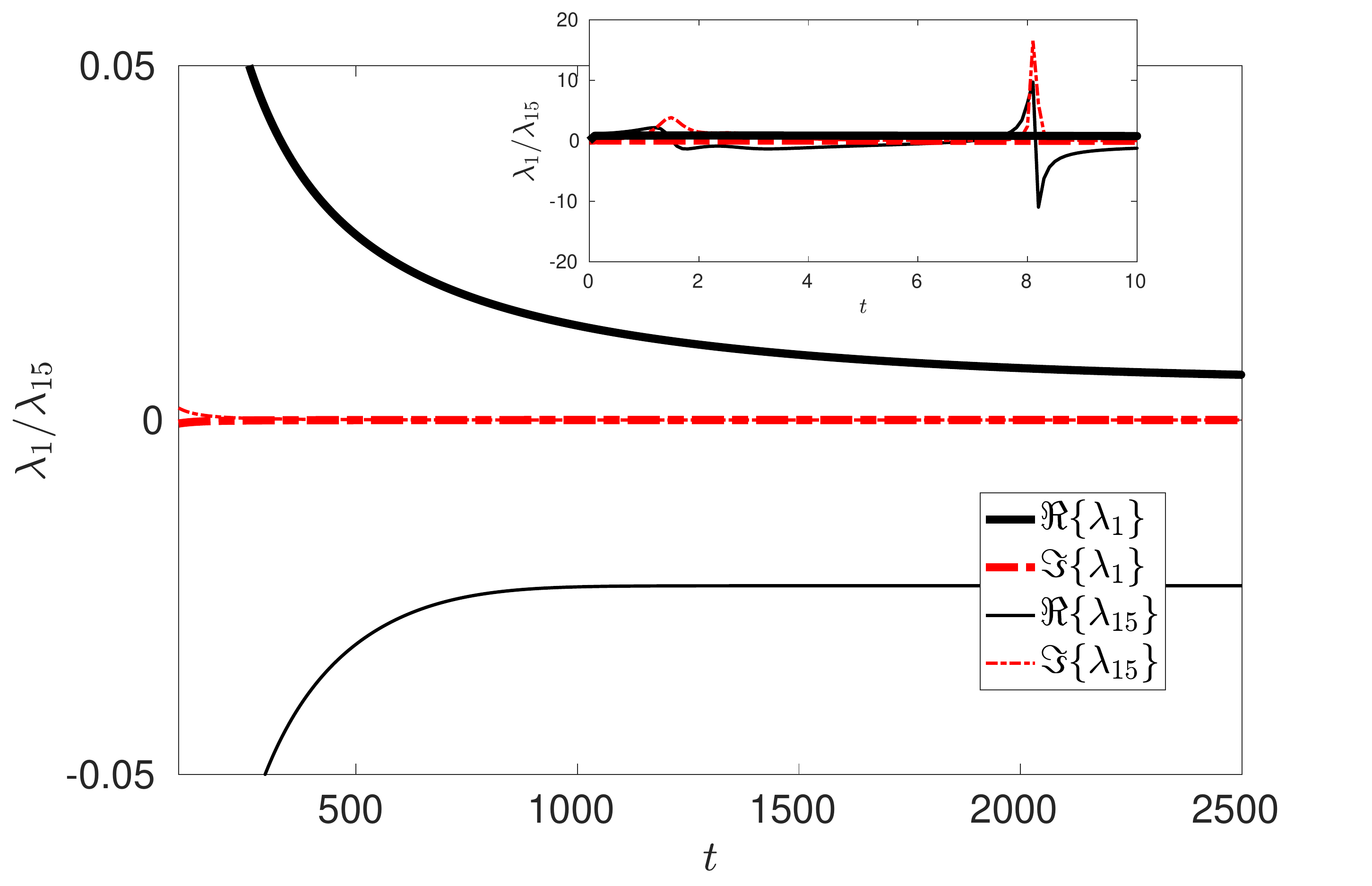}\label{fig:Hess_dyna}}
		\hspace*{1cm}\subfigure[Experiment E2: convergence of real and imaginary parts of the eigenvalues having the largest  ($\lambda_1$) and the smallest ($\lambda_{15}$) real part obtained with \eqref{32}
		for the Hessenberg matrix $D\{1/(k^2!)\}$ of order $n=15$. In the inset the early time evolution.]{\includegraphics[width=8.5cm]{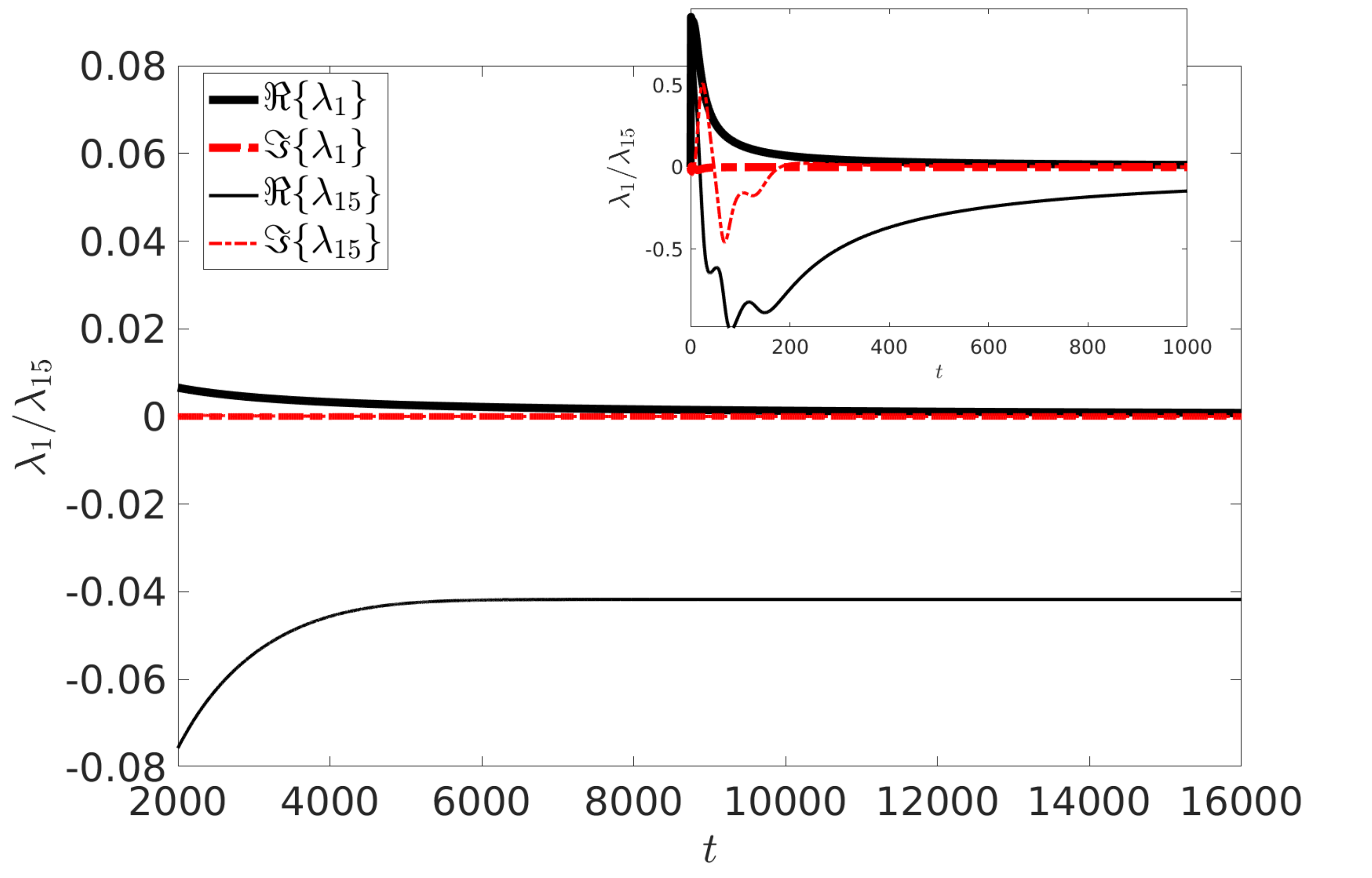}\label{fig:Hess_dynb}}
			\caption{Numerical results of dynamical approach  applied to the experiment E2. }	
	\end{center}
\end{figure}

\subsubsection{Fixed point approach, experiment E2}
In this subsection we consider the same matrix for the same sequences and the same tolerances, but we use the fixed point strategy. 
Convergence of the three largest eigenvalues in norm are shown in Figs.\ref{fig:Hess_fpa}-\ref{fig:Hess_fpnb} (only real part are shown). The eigenvalues we find are $\lambda_1=-0.0233,\lambda_2=0.0059,\lambda_3=-0.00069$ for the matrix $D\{\exp(-k^2)\}$ of order $n=15$ , and $\lambda_1=-0.0417,\lambda_2=6.58\cdot 10^{-5},\lambda_3=1.81\cdot 10^{-8}$ for the 
matrix $D\{1/(k^2!)\}$ of order $n=15$.

\begin{figure}[!ht]
	\begin{center}	
		\hspace*{-2cm}\subfigure[Experiment E2: convergence of real parts of the three largest in norm eigenvalues for the Hessenberg matrix $D\{\exp(-k^2)\}$ of order $n=15$ with the fixed point strategy.]
		{\includegraphics[width=8.5cm]{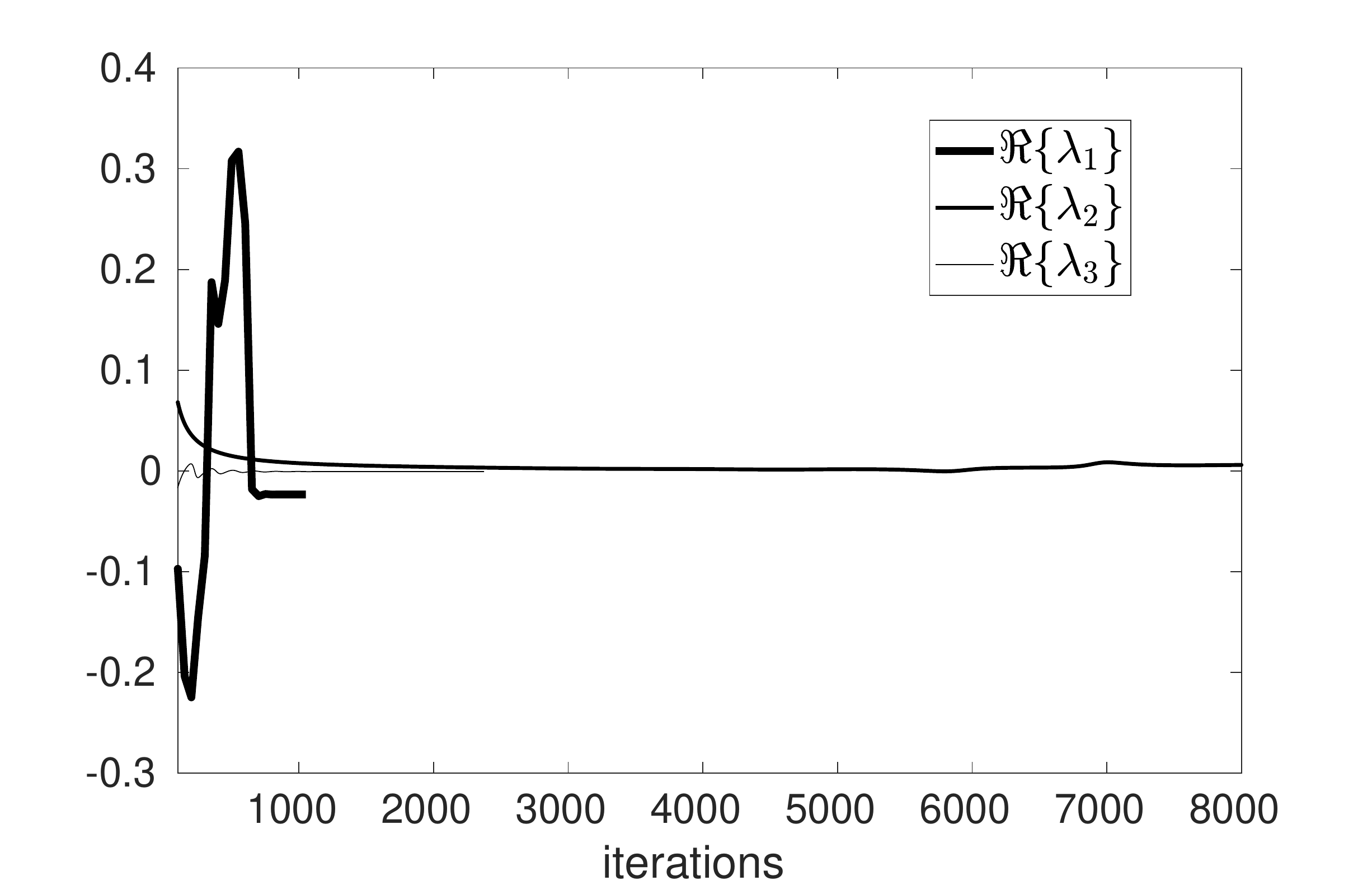}\label{fig:Hess_fpa}}
		\hspace*{1cm}\subfigure[Experiment E2: convergence of real parts of the three largest in norm eigenvalues for
		 the Hessenberg matrix$D\{1/(k^2!)\}$ of order $n=15$ with the fixed point strategy. ]{\includegraphics[width=8.5cm]{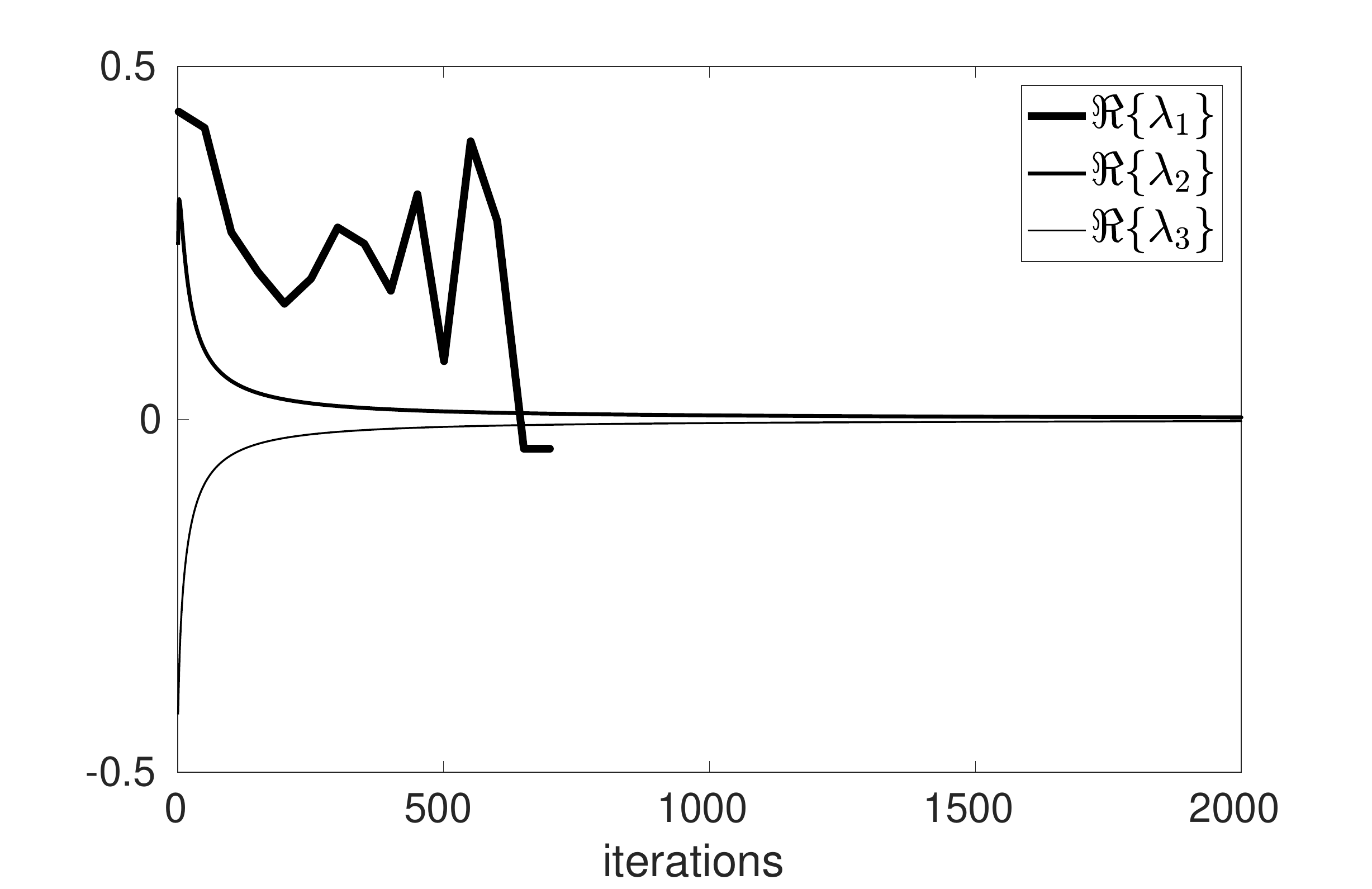}\label{fig:Hess_fpnb}}	
		\caption{Numerical results of fixed point strategy  applied to the experiment E2. }	
	\end{center}
\end{figure}

\subsection{Application to Quantum Mechanics: the truncated Swanson model}
In this third numerical experiment (E3), we work with a finite matrix that has a relevance in the contest of
pseudo-Hermitian Quantum Mechanics. In particular we consider the truncated Swanson model (hereafter TSM),
characterized by a finite Hamiltonian matrix which is not self adjoint, but which still admits only  real eigenvalues.
In the following we will briefly recall how this Hamiltonian can be obtained, while we refer the interested reader to \cite{bag2018} for more details and, in particular, for the physical relevance of this model.

The TSM Hamiltonian can be written as 
$$
H_{\theta}=\frac{1}{\cos(2\theta)}\left(B_\theta A_\theta+\frac{1}{2}\left(\1-Nk\right)\right),
$$
where $N$ is a non negative integer fixing the dimension of the system, 
$\theta\in(-\pi/4,\pi/4)\backslash\{0\}$ is a parameter tuning the non Hermiticity of the system,
$A_\theta$ and $B_\theta$ are operators (matrices) satisfying the commutation rule
$$
\left(A_\theta B_\theta -B_\theta A_\theta\right) f=f-Nkf,\quad \forall f\in \mathbb{C}^N,
$$ 
and $k$ is a projection operator
which annihilates the vector $e_N$ of the canonical orthonormal basis of $\mathbb{C}^N$, and satisfies  $k = k^
2 = k^\dag$ together with $kA_\theta=B_\theta k=0$. 

In \cite{bag2018} it is shown that $H_\theta$ is similar to the truncated quantum harmonic oscillator Hamiltonian 
$h$,
\be
H_\theta=T_\theta h T_\theta^{-1},
\label{add4}\en
where  $T_\theta=\exp\left(i\theta(a^2-(a^\dagger)^2)\right)$ and $a$ is the truncated annihilator operator defined in \cite{batrunc}, which satisfies the ladder equations on the basis $\{e_k\}_{k=1,\ldots N}$:
$$
ae_1=0,\quad a e_k=\sqrt{k}e_{k-1},\quad a^\dag e_k=\sqrt{k+1}e_{k+1}, \quad a^\dag e_N=0.
$$

Equation (\ref{add4}) implies that $H_\theta$ has the same spectrum as $h$, $\mu_k=\frac{2(k-1)+1}{2}, k=1,\ldots,N.$\footnote{Notice that maintaining the same formalism of the previous sections, the numerical eigenvalues, ordered from the one with the largest real part to the lowest, are labelled as $\lambda_1,\ldots,\lambda_N$, so that $\lambda_k$ will correspond to $\mu_{N+1-k}$, $k=1,\ldots,N$. }
For concreteness, we fix now $N=7$ and  $\theta=0.4$. Then the TSM Hamiltonian has the following form:
$$
H_\theta=\left(
\begin{array}{ccccccc}
 0.348126 & 0 & -0.510541 i & 0 & 0.0140773 & 0 &  0.0216558 i \\
 0 & 1.05673 & 0 & -0.805139 i & 0 & -0.145695 & 0 \\
 -0.510541 i & 0 & 1.79157 & 0 &  -1.01756 i & 0 & -0.372785 \\
 0 & -0.805139 i & 0 & 1.9337 & 0 &  -2.76093 i & 0 \\
 0.0140773 & 0 &  -1.01756 i & 0 & 2.0337 & 0 &  -4.04439 i \\
 0 & -0.145695 & 0 & -2.76093 i & 0 & 7.50957 & 0 \\
  0.0216558 i & 0 & -0.372785 & 0 &  -4.04439 i & 0 & 9.8266 \\
\end{array}
\right)
$$
\subsubsection{Dynamical approach, experiment E3}
We solve \eqref{32bis} for the matrices $H_\theta$ and $H_\theta^\dag$ with the tolerance
$\delta_{tol}=10^{-8}$. Starting with  random initial conditions, the solutions $x_{\varphi_1}^\infty$ and $x_{\Psi_1}^\infty$  related to the first eigenvalues $\lambda_1$, are
  $$x_{\varphi_1}^\infty=\left( \begin{array}{l}
  -0.0088 + 0.0149i\\
  -0.0000003 + 0.000001i\\
  -0.1676 - 0.0983i\\
   0.0000001 - 0.0000001i\\
   0.3206 - 0.5465i\\
   0.0000001 + 0.0000002i\\
   0.6457 + 0.3789i
   \end{array}\right),\quad 
   x_{\Psi_1}^\infty=\left( \begin{array}{l}
        0.00014 - 0.000243444797343i\\
        -0.00000136 - 0.000000641i\\
        -0.0027 - 0.0016i\\
        -0.00000028 + 0.00000065i\\
        -0.0052 + 0.0089i\\
         0.0000082 + 0.00000038i\\
         0.0105 + 0.0061
      \end{array}\right).
   $$

In Figure \ref{fig:7x7swaa} the convergence of the absolute values of the components of the dynamical solution $x_{\varphi_1}(t)$ converging to  $x_{\varphi_1}^\infty$  is shown (we get a similar behaviour for the convergence to $x_{\Psi_1}^\infty$, not shown  in figure). Real and imaginary parts of the eigenvalue $\lambda_1$ obtained from 
   \eqref{332}, are shown in
   Figure \ref{fig:7x7swac}, and as expected they both converge to the eigenvalue of $H_\theta$ with the largest real part, that is $\mu_7=6.5$. Considering the results related to the lowest eigenvalue, $\mu_1=0.5$, retrieved by switching to $-H_\theta$ and $-H_\theta^\dag$, the  converging solutions $x_{\varphi_7}^\infty$ and $x_{\Psi_7}^\infty$ are
     $$x_{\varphi_7}^\infty=\left( \begin{array}{l}
    0.8363 - 0.4557i\\
       0.0000002 - 0.0000003i\\
       0.1362 + 0.2499002i\\
       0.0000004 + 0.0000i\\
      -0.0908 + 0.0495i\\
      -0.0000003 + 0.0000001i\\
      -0.0171 - 0.0313i\\
      \end{array}\right),\quad 
      x_{\Psi_7}^\infty=\left( \begin{array}{l}
     -0.6568 + 0.3579i\\
      -0.0000001 + 0.00000002i\\
       0.1070 + 0.1963i\\
       0.0000004 + 0.00000003i\\
       0.0713 - 0.0389i\\
       0.0000005 - 0.00000001i\\
      -0.0134 - 0.0246i\\
         \end{array}\right).
      $$
   Convergence of the absolute values of the related dynamical solution $x_{\varphi_7}(t)$ to $x_{\varphi_7}^\infty$  is shown in \ref{fig:7x7swab}, whereas the real and imaginary parts of the corresponding eigenvalue $\lambda_7$ obtained from 
   \eqref{332} are shown in
   Figure \ref{fig:7x7swac}.
\begin{figure}[!ht]
	\begin{center}	
		\hspace*{-1.5cm}\subfigure[Experiment E3: time evolutions of the components of $|x_{\varphi_1}(t)|$ obtained from \eqref{32}]{\includegraphics[width=8.5cm]{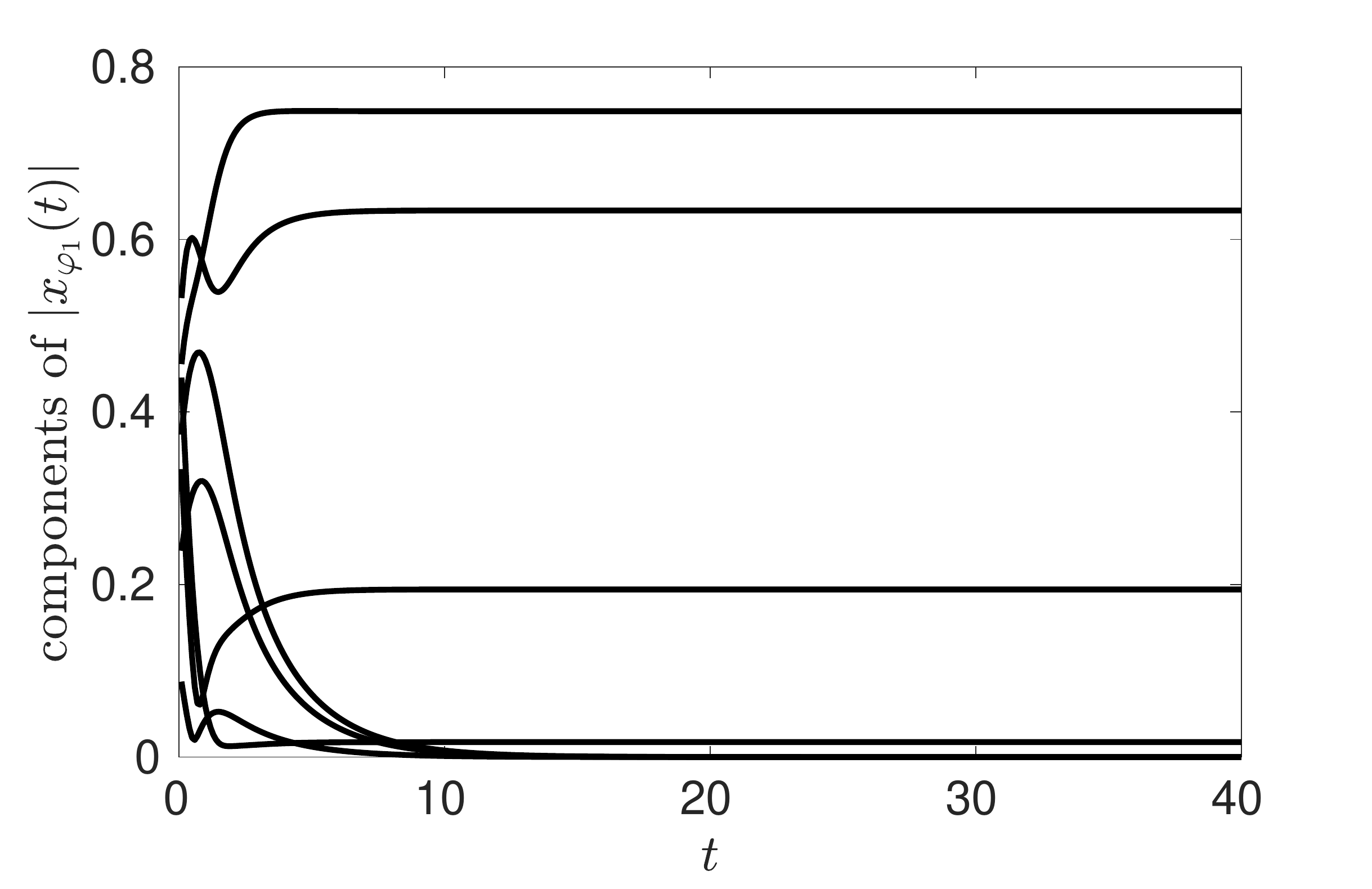}\label{fig:7x7swaa}}
			\hspace*{1cm}\subfigure[Experiment E3: time evolutions of the components of $|x_{\varphi_7}(t)|$ obtained from \eqref{32}]{\includegraphics[width=8.5cm]{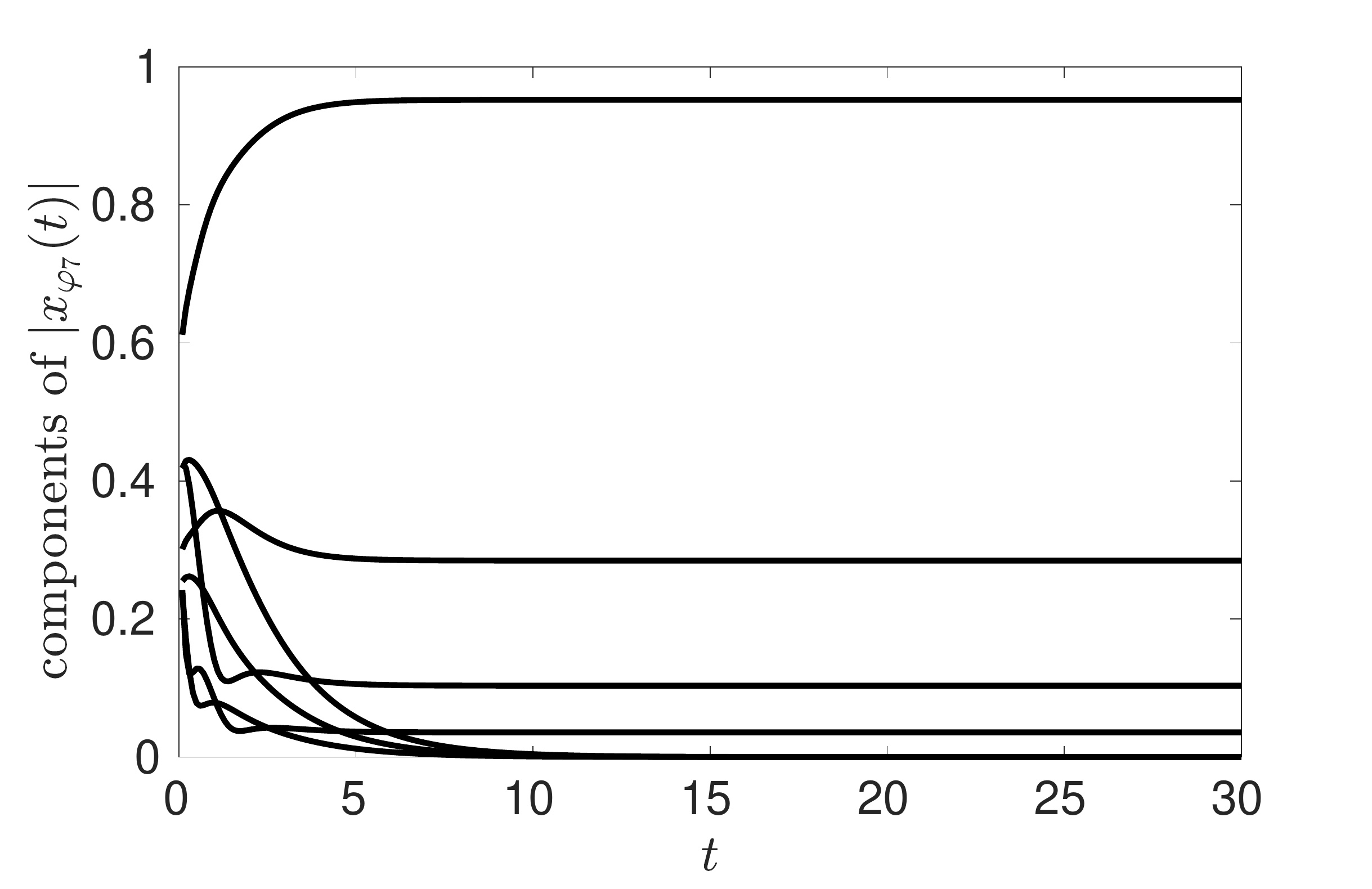}\label{fig:7x7swab}}	
	\subfigure[Experiment E3: time evolutions of the real and imaginary parts of the eigenvalues $\lambda_1$ and $\lambda_7$ of $A$]{\includegraphics[width=8.5cm]{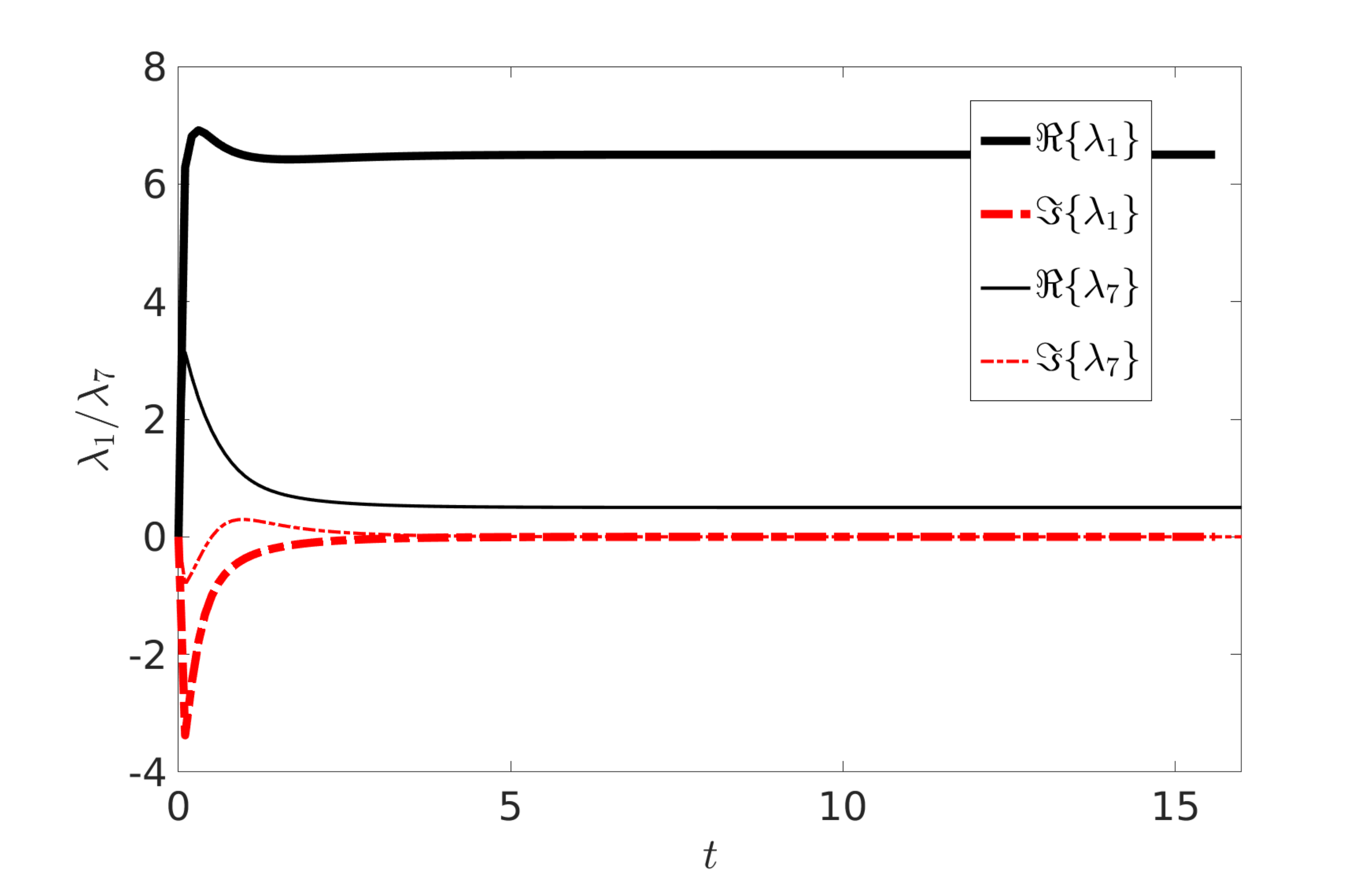}\label{fig:7x7swac}}
		\caption{Numerical results of dynamical approach applied to the experiment E3. }
		
	\end{center}
\end{figure}

\subsubsection{Fixed point approach, experiment E3}
In this subsection we show the results concerning the fixed point strategy applied to  $H_\theta$, again  with $N=7$, $\theta=0.4$ and a tolerance $\delta_{tol}=10^{-8}$. 
We report in Figure \ref{fig:SWAfp} (for simplicity only the real parts are shown) the convergence of the seven eigenvalues after the applications of the shifted inverse power method. Of course the number of iterations needed for convergence is highly sensitive to  the randomly value $q$ used to generate the inverse matrix $A-q\1$, as it is clearly shown in Figure \ref{fig:SWAfp}: some eigenvalue is found after just few iterations, while others need more iterations to be reached.

\begin{figure}[!ht]
\begin{center}
		\subfigure[Experiment E3: time evolutions of the real  part of the 7 computed eigenvalues  $\lambda_1,\ldots,\lambda_7$.]{\includegraphics[width=8.0cm]{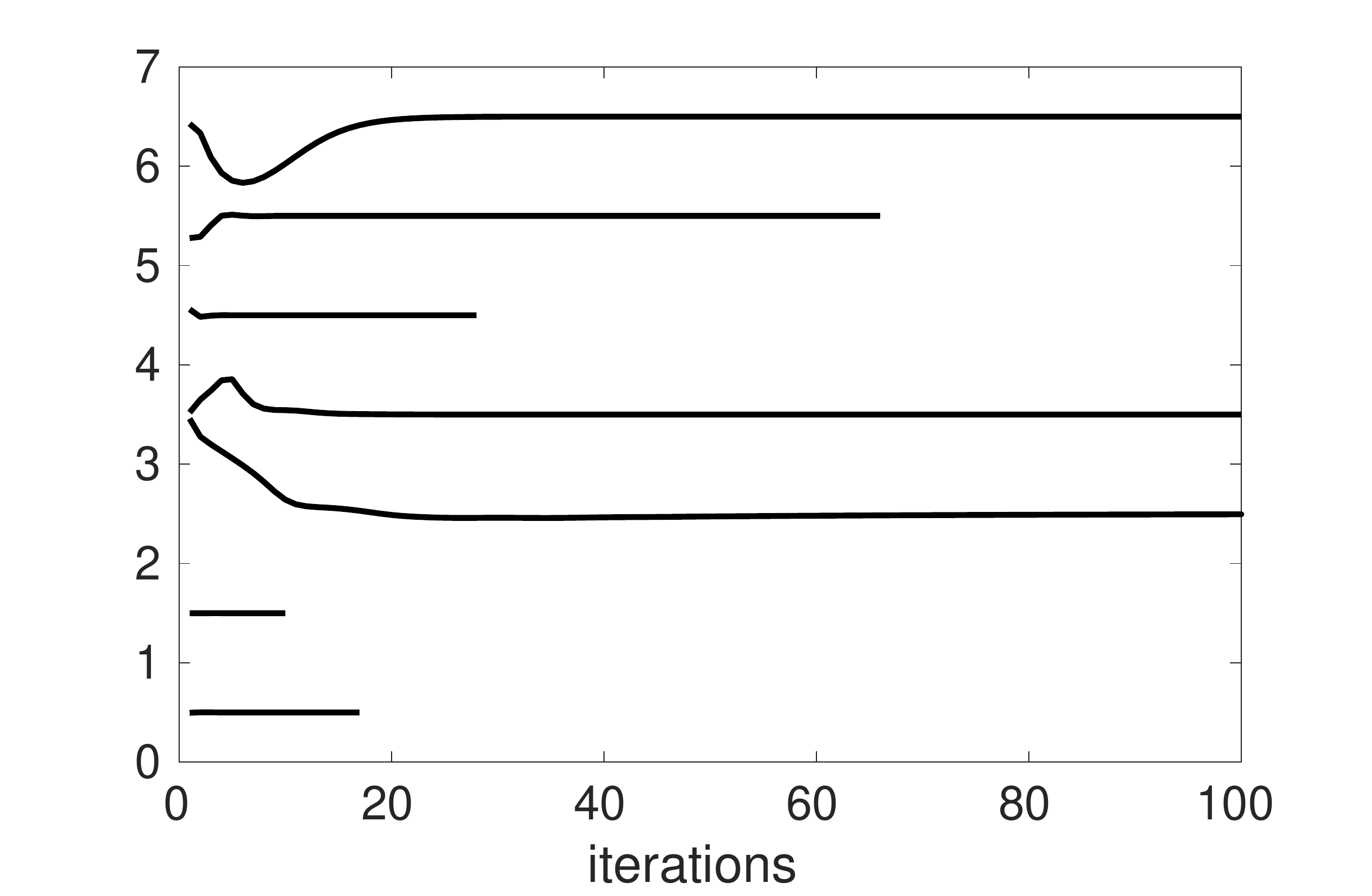}}
		\caption{Numerical results of fixed point approach applied to the experiment E3. }
		\label{fig:SWAfp}
		\end{center}
\end{figure}

\section{Conclusions}\label{sect5}

In this paper we have discussed how some standard techniques existing in the literature to compute eigenvalues and eigenvectors of a given Hermitian matrix can be extended to include matrices which are not Hermitian. { In particular we have considered a dynamical approach based on the solution of a system of ODEs, which naturally extends the procedure proposed in \cite{tang}, and a fixed point approach based on the construction of a suitable contraction as in \cite{schae}.  Many scientific and data analysis applications require the determination of the eigensystem of Hermitian matrices, while  in this work our main interest is moved to non Hermitian matrices. This kind of matrices, if seen as bounded (or even unbounded, if the matrices are infinite) operators, are quite often met in pseudo-hermitian quantum mechanics, where the Hamiltonian of a given system is not required to be Hermitian at all, but rather to satisfy some invariance property, \cite{benbook}.}
In this case, it is very likely that the eigenvalues of the Hamiltonian are complex and the mathematical and numerical setting proposed in \cite{tang} and in \cite{schae} should be adjusted to consider the appearance of biorthogonal sets of eigenvectors. 
{In this work we have extended the mathematical procedures presented in the aforementioned papers, and we have applied them to two pedagogical examples generated by a random matrix and by an Hessenberg matrix, and to an Hamiltonian operator obtained from a  finite-dimensional version of the Swanson model.
Our methodologies worked very well, and we were able to determine easily the eigensystem of the various matrices.} We should also mention that, as it is well known, when a given matrix possesses some symmetry, its eigenvalues obey some special rule. For instance, if $H$ is a matrix which is similar to a Hermitian operator $H_0$, i.e. if $H=SH_0S^{-1}$ for some invertible matrix $S$, then the eigenvalues of $H$ are all real. This is what happens, for instance, in PT- or pseudo-Hermitian quantum mechanics, \cite{benbook,op3}, and in our Experiment 3. Symmetries have played no role in this paper, but hopefully they will be considered in some future work.

{We stress that our work is mainly intended to provide some rigorous mathematical framework to deal with the eigenvalue problems considered in this paper, and in general for non Hermitian matrices, without focusing on any possible acceleration methods.} Of course,
modern information processing requires the solution of eigenvalue problems for very large matrices, and hence a consistent but also fast method is required.  A lot can still be done in trying to accelerate the convergence of the methods considered here.  This could involve the application of suitable acceleration procedures to speed up the convergence of the iterations, and this is just a part of our future works. Another relevant extension of our results should include those situations for which the convergences of our approaches fail. This is the case, for instance, when the real parts (resp. the norms) of the eigenvalues coincide, see Section \ref{sect3} (resp. Section \ref{sec:fixedpoint}).
We are  planning to consider a possible extension of the dynamical approach adopting (or extending) the fractal variational principle already used in literature for a quite wide class of applications (\cite{frac1,frac2,frac3}).
Also, we are interested in extending our ideas to infinite-dimensional matrices, and to compare our results with those in \cite{laub}. This is particularly interesting for us, in view of our interest for solving Schr\"odinger equations for concrete systems living in infinite-dimensional Hilbert spaces.

\section*{Acknowledgements}

This work was partially supported by the University of Palermo and by the Gruppo Nazionale di Fisica Matematica of Indam.
The authors want to thank Dr. Dario Armanno and the unknown Reviewers  for the useful comments and discussions which helped  to improve the final version of the paper.

\end{document}